\documentclass{IEEEoj}
\usepackage{cite}
\usepackage{orcidlink}
\usepackage{amsmath,amssymb,amsfonts}
\usepackage{algorithmic}
\usepackage{graphicx,color}
\usepackage{textcomp}
\def\BibTeX{{\rm B\kern-.05em{\sc i\kern-.025em b}\kern-.08em
    T\kern-.1667em\lower.7ex\hbox{E}\kern-.125emX}}
\AtBeginDocument{\definecolor{ojcolor}{cmyk}{0.93,0.59,0.15,0.02}}
\def\OJlogo{\vspace{-14pt}\includegraphics[height=28pt]{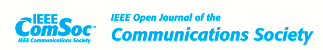}}

\usepackage{soul}

\usepackage{textcomp}
\usepackage{verbatim}

\usepackage{amsthm}
\usepackage{pgfplots}
\usetikzlibrary{spy}
\pgfplotsset{compat=1.18}

\usepackage{hyperref}
\hypersetup{
  colorlinks=false,
  linkbordercolor=white,
 urlbordercolor=white,
pdfborder={0 0 0}
}
\usepackage{makecell}
\usepackage{tikz-3dplot}
\usetikzlibrary{patterns}
\usetikzlibrary{arrows}
\usetikzlibrary{arrows.meta}
\usepackage{hhline}
\usetikzlibrary{calc}
\usetikzlibrary[pgfplots.groupplots]
\newtheorem{corollary}{Corollary}[section]
\newtheorem{lemma}{Lemma}[section]
\newtheorem{lemmaappendix}{Lemma}[section]
\newtheorem{definitionappendix}{Definition}[section]
\newtheorem{proposition}{Proposition}[section]

\newtheorem{observation}{Key Observation}[section]
\newtheorem{definition}{Definition}[section]

\begin{document}
\receiveddate{XX Month, XXXX}
\reviseddate{XX Month, XXXX}
\accepteddate{XX Month, XXXX}
\publisheddate{XX Month, XXXX}
\currentdate{XX Month, XXXX}
\doiinfo{OJCOMS.2022.1234567}

\title{Optimal Transmit Antenna Deployment and Power Allocation for Wireless Power Supply in an Indoor Space}

\author{KENNETH M. MAYER\orcidlink{0000-0003-3292-6038} (Graduate Student Member, IEEE),\\ LAURA COTTATELLUCCI\orcidlink{0000-0002-6641-8579} (Member, IEEE), AND ROBERT SCHOBER\orcidlink{0000-0002-6420-4884}
(Fellow, IEEE)}
\affil{Institute for Digital Communications, Department of Electrical Engineering, Friedrich-Alexander University Erlangen-Nuremberg, Germany.}
\corresp{CORRESPONDING AUTHOR: Kenneth M. Mayer (e-mail: kenneth.m.mayer@fau.de).}
\authornote{This work was (partly) funded by the Deutsche Forschungsgemeinschaft (DFG, German Research Foundation) – SFB 1483 – Project-ID 442419336, EmpkinS.}
\markboth{Optimal Transmit Antenna Deployment and Power Allocation for Wireless Power Supply in an Indoor Space}{Mayer \textit{et al.}}

\begin{abstract}
As Internet of Things (IoT) devices proliferate, sustainable methods for powering them are becoming indispensable. 
The wireless provision of power enables battery-free operation and is crucial for complying with weight and size restrictions. 
For the energy harvesting {\color{black}(EH)} components of these devices to be small, a high operating frequency is necessary. 
In conjunction with a large {\color{black}transmit} antenna, the receivers may be located in the radiating near-field (Fresnel) region, e.g., in indoor scenarios.
In this paper, we propose a wireless power transfer {\color{black} (WPT)} system ensuring reliable supply of power to an arbitrary number of mobile, low-power, and single-antenna receivers, {\color{black}whose locations} in a three-dimensional cuboid room {\color{black}are unknown}.
%{\color{blue} Hereby, no knowledge regarding the receivers' locations is assumed at the transmitter.}
{\color{black}A max-min optimisation problem is formulated to determine the optimal transmit power distribution.}
%To this end, we formulate a max-min optimisation problem to determine the optimal allocation of transmit power on a continuous array comprising an infinite{\color{blue}, uncountable} number of radiating elements{\color{blue}, such as, in a Holographic MIMO (HMIMO) system}.
{\color{black}
We rigorously prove} that the optimal transmit power distribution’s support has a lower dimensionality than its domain and thus, the employment of a continuous aperture antenna, utilised in Holographic MIMO (HMIMO), is unnecessary in the context of the considered WPT problem.
Indeed, deploying a discrete transmit antenna architecture, i.e., a transmit antenna array, is sufficient and our proposed solution provides the optimal transmit antenna deployment and power allocation.
%Specifically, we prove that the support of the optimal transmit power distribution has Lebesgue measure zero and the closure of the set has empty interior.}
Moreover, for a one-dimensional transmit antenna {\color{black}architecture, a finite number of transmit antennas is proven to be optimal.}
%, the support of transmit antenna positions is proven to be finite.
%{\color{blue} By rigorously proving that, in general, the optimal transmit power density is characterised by mass points, the employment of a HMIMO system is always unnecessary in the context of our WPT problem and .}
%{\color{blue} Thereby the optimal deployment is obtained implicitly. When non-zero transmit power is allocated to a dense set of radiating elements, a continuous surface of radiating elements is necessary for optimal power allocation. In contrast, when the power is allocated to a discrete number of points of increase, a discrete transmit antenna array suffices.}
%Moreover, for a one-dimensional transmit antenna {\color{blue}architecture}, the set of transmit antenna positions is proven to be finite.
%{\color{blue}Thus, the deployment of a transmit antenna {\color{blue}architecture} with a spatially continuous %aperture, as in HMIMO, is not necessary.}
The proposed optimal solution is validated through computer simulations. 
Our simulation results indicate that the optimal {\color{black}transmit antenna architecture} requires a finite number of transmit antennas and depends on the geometry of the environment and the dimensionality of the transmit antenna array. 
The robustness of the solution, which is obtained under the assumption that a line-of-sight (LoS) link exists between the transmitter and receiver, is assessed in an isotropic scattering environment containing a strong LoS component.
\end{abstract}

\begin{IEEEkeywords}
Wireless power transfer, optimisation, radiating near-field region.
\end{IEEEkeywords}


\maketitle

\section{INTRODUCTION}\label{Section: Introduction}
\IEEEPARstart{T}{he} number of Internet of Things (IoT) devices continues to grow rapidly and is expected to exceed 43 billion devices globally by 2023 \cite{Rahmani23}. 
Among those are wearable IoT devices, which are primarily intended for collecting data, e.g., for patient monitoring, treatment or rehabilitation \cite{Dian20,VERMA2022100153}.
Wearable IoT devices are limited in terms of their weight and size and have to satisfy aesthetic requirements \cite{Metcalf16}. 
For example, these smart devices are not only designed as wearable accessories, but also as implants or tattoos and they may be embedded into textiles \cite{Dian20,VERMA2022100153}. 

Sustainability is a central theme for designing next-generation IoT technology and includes powering IoT devices, such as wearable systems, health monitoring systems, and implants, sustainably \cite{Rahmani23}.
Currently, batteries are considered a hindrance in IoT technology due to the restrictions they impose on devices in terms of, e.g., size, weight, cost, and the need for replacement \cite{Wagih23,Wagih20,Rahmani23}.
Alternative technologies, such as wireless power transfer (WPT) and energy harvesting (EH), facilitate the operation of battery-free IoT devices \cite{Rahmani23,Wagih23,Wagih20}. 
A reliable supply of power to IoT devices {\color{black}can be} ensured through the combination of available ambient power sources, such as ambient light, and dedicated power sources \cite{Rahmani23,Wagih23}. 
Ambient sources alone can not ensure {\color{black}reliable powering} as they are often uncontrollable. 
In contrast, dedicated power sources, such as dedicated electromagnetic (EM) radiation, offer controllability and thus, can {\color{black}ensure that} IoT devices  {\color{black}are powered} adequately \cite{Wagih23}.

{\color{black}
The acquisition and use of channel state information (CSI) plays a fundamental role in wireless information and power transfer.
When the acquisition of CSI is feasible, beamforming-based approaches for directing energy beams towards a multitude of possibly mobile receiver devices are possible.
However, the resulting WPT system may be highly complex \cite{Rahmani23}, for example, the WPT system proposed in \cite{Lopez22} utilises CSI acquired through uplink training to power devices in the environment.}
{\color{black} Often, IoT devices are mobile which makes the acquisition of reliable CSI at the transmitter challenging. 
Moreover, certain types of IoT devices, such as wirelessly powered sensor nodes, may lack the communication capabilities to broadcast pilot symbols, which makes serving EH nodes with power a challenge since focusing energy beams towards them is not possible in a reliable manner due to the missing CSI at the energy transmitter.}
{\color{black}In this situation,} there is a significant risk that the energy beams are misdirected and fail to reach the receivers, leading to an insufficient supply of power to these devices.
To ensure the reliable provision of power to an arbitrary number of devices in an environment, the WPT system must be capable of supplying the served environment with a minimum guaranteed level of power at any location without relying on CSI \cite{Lopez21,Lopez21-2}. 
%{\color{black} We note that this timely problem has not been addressed for indoor spaces yet.}

%Although WPT systems are promising for powering IoT devices, their design must account for practical issues such as the acquisition of reliable channel state information (CSI).
%Consequently, CSI-free WPT approaches have been proposed in the context of powering IoT devices \cite{Clerckx21}.

{\color{black} 
An essential component for providing a ubiquitous supply of power in an indoor environment is the transmit antenna architecture.
Novel antenna architecture designs are being discussed as key enabling technologies for future wireless systems.
Advancements beyond massive multiple-input multiple-output (mMIMO) systems include increasing the physical size of antenna arrays and increasing the number of transmit antennas while decreasing the inter-element spacing \cite{Wang24}.
In fact, through the incorporation of novel meta-materials, an approximately infinitesimally small inter-element spacing is achieved, allowing to pack almost infinitely many radiating elements into an antenna array, thereby forming a spatially continuous antenna aperture \cite{Wang24}.
The technology revolving around these continuous aperture antennas has been coined holographic MIMO (HMIMO) and allows for, e.g., holographic beamforming \cite{Bjornson19, Huang20}.
HMIMO is being considered as an important technology, e.g., for future 6G wireless communication systems \cite{Bjornson19}.
Moreover, the HMIMO concept is envisioned to offer a seamless and visually appealing integration of antenna arrays into the surfaces of an environment including walls, the ceiling, and even windows \cite{Bjornson19, Huang20}.
While the benefits of HMIMO have been demonstrated in the context of communications, the importance of HMIMO for WPT remains unclear.}
%Naturally, the aesthetic requirements imposed on devices also apply to the surrounding environment, which facilitates the operation of the devices, such that the overall environment is visually pleasing.

{\color{black} 
    Indoor wireless systems have been ubiquitous for decades and thus, methods for accurately modelling the wireless channel have been studied extensively.
    In particular, the frequency range of 2.4 GHz to 5 GHz has received significant attention and thus, deterministic, statistical, and site-specific channel models have been proposed, which differ in accuracy and computational efficiency 
    \cite{Anusuya08}.}
{\color{black} In \cite{Abdulwahid19}, the propagation characteristics of indoor wireless systems operating in the C- and mmWave frequency bands have been investigated for line-of-sight (LoS) and non-line-of-sight (NLoS) scenarios via a software-based ray tracing model. 
Hereby, the increase in attenuation (path loss) upon increasing the frequency was found to be significantly more severe for the NLoS scenario.
Consequently, in the mmWave frequency band, a strong LoS link is required for reliable operation of indoor wireless systems, and due to the severely attenuated scattered components, a LoS channel model is typically adopted \cite{Liu16,Perovic20}.}
{\color{black}We note that} high-frequency systems are attractive for short-range applications, such as indoor WPT \cite{Zhanga}. The short wavelength allows the EH antenna of a device to be small \cite{Wagih20}. For example, when utilising mmWave frequency bands, mm-scale antennas can be employed \cite{Wagih20}. Consequently, the overall size of the devices can be reduced, allowing them to be less obtrusive and easier to integrate, e.g., into fabrics and textiles \cite{Wagih23,Rahmani23}.
Beyond that, a high carrier frequency in conjunction with large antenna {\color{black}architectures, such as, the continuous aperture antennas used in HMIMO,} leads to a significant increase in the radiating near-field (Fresnel) region. 
Although the Fresnel region could typically be neglected in conventional wireless systems, it has to be considered in, e.g., indoor scenarios, for high-frequency systems. Within the Fresnel region, the spherical nature of the EM wavefronts is non-negligible. Adequately capturing the EM wave characteristics in the Fresnel region demands modelling the wireless channel according to the spherical wavefront model (SWM), as the approximation by the planar wavefront model (PWM), typically utilised in the far-field region, does not hold \cite{Zhang2022, Hu2018, Dardari2021}.

%Naturally, the aesthetic requirements imposed on devices also apply to the surrounding environment, which facilitates the operation of the devices, such that the overall environment is visually pleasing.
%A concept envisioned to offer a seamless and visually appealing integration of antenna arrays into the surfaces of the environment, such as walls, the ceiling, and windows, is known as Holographic MIMO (HMIMO) \cite{Bjornson19, Huang20}.
%HMIMO systems employ transmit antenna arrays with spatially continuous apertures  \cite{Bjornson19, Huang20}.
%{\color{black} HMIMO is discussed as a key enabling technology for future wireless systems and allows for, e.g., holographic beamforming \cite{Bjornson19, Huang20}.}

%A WPT system operating in the Fresnel region is capable of focusing energy beams directly towards the receivers \cite{Zhanga}.
%{\color{black} Additionally, a HMIMO WPT system provides holographic beamforming capabilities \cite{An23}.}

Considering the rapidly increasing number of IoT devices \cite{Dian20,VERMA2022100153,Rahmani23}, the development of a WPT system, capable of reliably supplying them with power is an important and timely challenge. {\color{black} Emerging technologies, such as HMIMO systems operating in the mmWave frequency band, may play a key role in addressing this problem.}
{\color{black} Related works have studied the optimisation of the deployment of transmit antennas with the goal of supplying devices with power, e.g., in \cite{Rosabal21,Jia21,Zhou23,Tabassum15,Amirhossein23}.}
{\color{black}
%The optimisation of the deployment of transmit antennas with the goal of supplying devices with power has been investigated in \cite{Rosabal21,Jia21,Zhou23,Tabassum15,Amirhossein23}, where transmit antennas are referred to as power beacons.
{\color{black}In \cite{Rosabal21,Jia21,Zhou23,Tabassum15}, operation in the far-field regime is assumed, statistical channel models are adopted, and a coverage area, significantly larger than an indoor space, is considered.}
For these systems, statistical metrics, such as, outage probability, are utilised for optimising or characterising the positions of power beacons.
On the other hand, as discussed previously, a high operating frequency, such as, in the mmWave frequency band, is affordable and attractive for indoor WPT systems, where a deterministic LoS channel model appropriately characterises the propagation of EM waves. 
As a consequence of the deterministic channel model, the techniques proposed in \cite{Rosabal21,Jia21,Zhou23,Tabassum15} are not applicable and an alternative approach for optimising the deployment of the transmit antennas is required due to the absence of statistical effects.
Furthermore, the authors of \cite{Amirhossein23} optimise the location of a finite number of radio stripe elements through geometric programming to maximise the minimum power received by users in indoor energy hotspots in the radiating near-field region.
{\color{black}A continuous transmit antenna architecture, as in HMIMO, is precluded as a potential deployment solution when restricting the analysis to a pre-defined number of transmit antennas}.
Consequently, this leaves the question unanswered whether the novel continuous transmit architecture is beneficial in the context of offering a ubiquitous supply of power in an indoor environment.
We aim to address this open point in this work.}

In this paper, we propose a multi-user multiple-input single-output (MU-MISO) WPT system, which supplies an entire indoor space with a guaranteed minimum level of power. 
The {\color{black} transmit antenna architecture} of the WPT system {\color{black}is realised via} a {\color{black} continuous aperture antenna}, which is mounted on the ceiling of a three-dimensional cuboid room.
%, and  to form a spatially continuous aperture. } 
{\color{black} The proposed WPT system can serve an arbitrary number of single-antenna receivers.}
Assuming free space LoS propagation of the spherical EM wavefronts, we formulate an optimisation problem based on the distance information contained in the amplitude variations of the wireless links between every possible {\color{black}infinitesimally small radiating element and every possible receiver position to determine the optimal distribution of power across the continuous aperture antenna of the transmitter.} 
The objective is to maximise the received power at the worst possible receiver positions in the indoor space. Thus, the proposed solution ensures a reliable supply of power to the receivers, which is a relevant criterion when there is no prior knowledge on the receivers' locations in the room.
{\color{black} We establish a connection between the optimal distribution of transmit power and the optimal transmit antenna architecture to study the relevance of emerging HMIMO technology to WPT in an indoor environment.
Specifically, employing a continuous aperture antenna would be beneficial if the optimal transmit power distribution’s support was found to have the same dimensionality as the distribution’s domain.
%Specifically, employing a continuous aperture antenna would be beneficial if the optimal transmit power density was found to be continuous.
%However, if the optimal power density could be characterised by mass points, HMIMO technology was not beneficial.
However, if the optimal transmit power distribution’s support has a lower dimensionality than its domain, HMIMO technology was not beneficial.
%In this case, the optimal transmit power distribution’s support represents the optimal deployment locations of transmit antennas along with their respective optimal transmit powers.
}
%The optimal power allocated to {\color{black}radiating elements} can be zero, and thus, the optimisation problem for power distribution implicitly yields the optimal {\color{black} deployment of radiating elements .}
%{\color{black} In this case, it is optimal to deploy a transmit antenna where the optimal power density of transmit power has a mass point}.
{\color{black} The} adopted continuous {\color{black} aperture model} allows us to capture the full range of {\color{black}transmit antenna architectures} ranging from a single transmit antenna to a continuous {\color{black} surface} of radiating elements as in HMIMO, thereby enabling a complete analysis of the problem.
{\color{black} In contrast, the relevance of HMIMO technology to the considered problem cannot be adequately studied based on a discrete antenna model.}
%We note that assuming a discrete transmit antenna array model would preclude a continuous transmit antenna architecture, such as a HMIMO system, as a possible solution, thereby placing a restriction on the analysis.  
The main contributions of this paper can be summarised as follows: 
\begin{itemize}
\item {\color{black} We establish a method for investigating the relevance of continuous aperture antennas in the context of offering a ubiquitous supply of power in an indoor environment.
{\color{black}To this end, we formulate} an optimisation problem for the distribution of a finite amount of transmit power among an uncountable number of radiating elements with the objective of maximising the received signal power at the worst possible receiver position in a given environment. }
\item We prove that a continuous {\color{black}aperture} antenna is unnecessary {\color{black}for the considered problem} since our analysis reveals that the support of the optimal {\color{black} distribution} of transmit power is a set whose closure has empty interior and Lebesgue measure zero{\color{black}, i.e., the support is shown to have a lower dimensionality than the domain of the transmit power distribution.} 
%{\color{black}Consequently, a discrete transmit architecture is sufficient for }
%This holds for a set of isolated points. However, more complex sets with measure zero could in theory also be possible.
The support of the optimal {\color{black} distribution} corresponds to the optimal positions where transmit antennas should be deployed and {\color{black} we conjecture that the optimal transmit antenna architecture can be realised through a discrete antenna array}.
In case the transmit antenna {\color{black} architecture} is constrained to be one-dimensional, a stronger statement is possible and the optimal number of radiating elements is shown to be finite analytically.
\item Through simulation we verify that the optimal {\color{black}transmit power distribution's support is characterised by a finite number of points, and thus, the optimal deployment consists of a finite} number of transmit antennas for {\color{black}both} one-dimensional and two-dimensional transmit antenna {\color{black} architectures}. 
%The actual number of antennas needed {\color{black}in the distributed antenna system} depends on the geometry of the environment. 
In the investigated environments, the optimal number {\color{black} of transmit antennas} is found to be surprisingly low.
\item To validate our findings, the performance of the proposed optimal solution is evaluated against other power allocation schemes. 
{\color{black} Specifically, we investigate the achievable performance gain of modelling the system in the Fresnel region 
%by comparing the proposed approach to the optimal power distribution designed for the far-field operating regime.
%Hereby, 
and we observe that the performance gap compared to the optimal distribution for the far-field operating regime decreases when the worst possible receiver positions in the environment are almost in the far-field. 
We also show that the loss in performance is low when removing those transmit antennas from the optimal transmit antenna array that are allocated a small amount of transmit power.}
%We compare our solution to uniform distribution of transmit power.}
\item {\color{black}The robustness of the proposed solution is investigated by adding small-scale fading to the LoS link, which leads to a performance loss of less than $2\%$ in the considered environments.}
\end{itemize}
   
The paper is organised in the following manner. 
In Section \ref{Section: System Model}, the system model is introduced.
In Section \ref{section: Methods}, the proposed optimisation problem is formulated and the structure of the optimal {\color{black}distribution} of power is determined. 
As a result, {\color{black}we prove that a finite number of transmit antennas is optimal for the one-dimensional problem and} we conjecture that {\color{black}a finite number of transmit antennas is also optimal regarding the two-dimensional problem.}
{\color{black}In Section \ref{Section: System Discretisation}, we discretise the problem for numerical simulation and, in support of our conjecture, we verify that the optimal number of transmit antennas remains unchanged while increasing the granularity of the sampling gird, i.e., for finer levels of discretisation.}
%can be obtained by discretising the problem, and subsequently, solving it numerically; this is done in Section \ref{Section: System Discretisation}.
In Section \ref{Section: Simulation Results}, we provide a performance and sensitivity analysis of our approach, and Section \ref{Section: Conclusion} concludes this paper.

\section{SYSTEM MODEL}\label{Section: System Model}
The {\color{black}WPT} system considered in this paper consists of a transmit antenna {\color{black}architecture employing a continuous aperture antenna} and an arbitrary{\color{black}\footnote{{\color{black} In beamforming-based approaches, such as \cite{Lopez22}, the number of supported receivers is limited by the number of transmit antennas. In our work, we do not face this restriction since we do not rely on beamforming.}}} number of single-antenna receivers. 
Throughout this paper, the terms "transmit antenna" and "radiating element" are used synonymously.
An illustration of the considered system is given in Fig. \ref{fig: System Model Overview}, where the transmit antenna {\color{black}architecture} is incorporated in the ceiling of a room.

\begin{figure}[t]
    \centering
    \hspace*{-4mm}\scalebox{0.8}{%%% Warning %%%
%%% x and y are swapped !!! %%%
%%% Warning %%%

\definecolor{FAU_BLAU}{RGB}{0,32,96}
\definecolor{ORANGE}{RGB}{255,165,0}
\tdplotsetmaincoords{70}{130}
\begin{tikzpicture}[
tdplot_main_coords,
cube/.style={very thick,black},
cube_continued/.style={thick,black},
indicator/.style={<->,thick,arrows = {Stealth[length=10pt, inset=2pt]-Stealth[length=10pt, inset=2pt]},thick,black},
grid/.style={very thin,gray},
axis_outside/.style={->,black,thick},
axis_inside/.style={very thin,opacity=0.33}
]
% Coordinates
\coordinate (wall_label) at (-4,2,0);
\coordinate (room_label) at (-3,1,1);
\coordinate (ceiling_label) at (-1,2,4);
\coordinate (obstructions) at (2,-1.5,4);
\coordinate (origin) at (2.5,2.5,4);
\coordinate (RX) at (2.75,2,1);

%draw the axes
\draw[axis_inside] (0,2.5,4) -- (5,2.5,4) ;
\draw[axis_outside] (5,2.5,4) -- (6,2.5,4) node[anchor=north]{$x$};
\draw[axis_inside] (2.5,0,4) -- (2.5,5,4) ;
\draw[axis_outside] (2.5,5,4) -- (2.5,6,4) node[anchor=west]{$z$};
\draw[axis_inside] (2.5,2.5,4) -- (2.5,2.5,0) ;
\draw[axis_outside] (2.5,2.5,0) -- (2.5,2.5,-1) node[anchor=south east]{$y$};

%dimensions
\draw[indicator] (5.6,0,0) -- (5.6,5,0) node[anchor=east,midway,yshift=-1mm]{$L_z$} ;
\draw[cube_continued] (5,0,0) -- (6,0,0) ;
\draw[cube_continued] (5,5,0) -- (6,5,0) ;

\draw[indicator] (0,5.6,0) -- (5,5.6,0) node[anchor=west,midway,yshift=-1mm]{$L_x$} ;
\draw[cube_continued] (5,5,0) -- (5,6,0) ;
\draw[cube_continued] (0,5,0) -- (0,6,0) ;

\draw[indicator] (5.6,0,0) -- (5.6,0,4) node[anchor=east,midway]{$L_y$} ;
\draw[cube_continued] (5,0,4) -- (6,0,4) ;

%draw the top and bottom of the cube
\draw[cube] (0,0,0) -- (0,5,0) -- (5,5,0) -- (5,0,0) -- cycle;
\fill[brown,cube,opacity=0.08] (0,0,0) -- (0,5,0) -- (5,5,0) -- (5,0,0) -- cycle;
\draw[cube] (0,0,4) -- (0,5,4) -- (5,5,4) -- (5,0,4) -- cycle;
\draw[cube] (1,1,4) -- (1,4,4) -- (4,4,4) -- (4,1,4) -- cycle;
\fill[orange,opacity=0.2] (1,1,4) -- (1,4,4) -- (4,4,4) -- (4,1,4) -- cycle;

%draw triangles
\fill[orange,opacity=0.08] (1,1,4) -- (RX) -- (4,1,4) -- cycle;
\fill[orange,opacity=0.08] (4,1,4) -- (RX) -- (4,4,4) -- cycle;
\fill[orange,opacity=0.08] (4,4,4) -- (RX) -- (1,4,4) -- cycle;
\fill[orange,opacity=0.08] (1,4,4) -- (RX) -- (1,1,4) -- cycle;

%draw the edges of the cube
\draw[cube] (0,0,0) -- (0,0,4);
\draw[cube] (0,5,0) -- (0,5,4);
\draw[cube] (5,0,0) -- (5,0,4);
\draw[cube] (5,5,0) -- (5,5,4);

% door 1
\draw[cube] (0,2.5,0) -- (0,3.5,0) -- (0,3.5,2) -- (0,2.5,2) -- cycle;
\fill[brown,cube,opacity=0.6] (0,2.5,0) -- (0,3.5,0) -- (0,3.5,2) -- (0,2.5,2) -- cycle;
\fill[yellow] (0,2.7,1) circle (1pt);

% door 2
\draw[cube] (3.8,0,0) -- (4.8,0,0) -- (4.8,0,2) -- (3.8,0,2) -- cycle;
\fill[brown,cube,opacity=0.6] (3.8,0,0) -- (4.8,0,0) -- (4.8,0,2) -- (3.8,0,2) -- cycle;
\fill[yellow] (4.6,0,1) circle (1pt);

% window 1
\draw[cube] (1.5,0,2.7) -- (1.5,0,3.5) -- (3.5,0,3.5) -- (3.5,0,2.7) -- cycle;
\fill[blue,opacity=0.2] (1.5,0,2.7) -- (1.5,0,3.5) -- (3.5,0,3.5) -- (3.5,0,2.7) -- cycle;
\draw[cube] (2.5,0,2.7) -- (2.5,0,3.5);

% window 2
\draw[cube] (0,1,0) -- (0,1.5,0) -- (0,1.5,3) -- (0,1,3) -- cycle;
\fill[blue,opacity=0.2] (0,1,0) -- (0,1.5,0) -- (0,1.5,3) -- (0,1,3) -- cycle;

\fill (origin) circle (1pt);
\fill[orange,opacity=0.9] (RX) circle (3pt) node[below left] {RX};
% \fill[orange] (a_1) circle (2pt) node[above ] {$a_1$};
% \fill[orange] (a_2) circle (2pt) node[above ] {$a_2$};

% Labels
\draw[o-] (-1,3.5,1) -- (wall_label) node[right] {walls};
\draw[o-] (1.5,6,4) -- (wall_label) ;
\draw[o-] (-1.25,2.5,-1) -- (wall_label) ;
%\draw[o-] (-1,2,1.5) -- (room_label) node[right] {room};
\draw[o-] (-0.5,1.2,3) -- (ceiling_label) node[right] { continuous aperture antenna};
\draw[o-] (0.35,1.7,2) -- (obstructions) node[above] {obstructions};
\draw[o-] (4.25,0,1.7) -- (obstructions);
\draw[o-] (2.7,0,3) -- (obstructions) ;

%\end{axis}
\end{tikzpicture}}
    \caption{Illustration of a three-dimensional environment with a two-dimensional transmit antenna {\color{black}architecture} consisting of a continuum of infinitesimally small {\color{black}radiating} elements located in a subsection of the $x$-$z$ plane (orange). The LoS connections of the transmit antenna elements to one receiver (RX) are indicated in orange.}
    \label{fig: System Model Overview}
\end{figure}
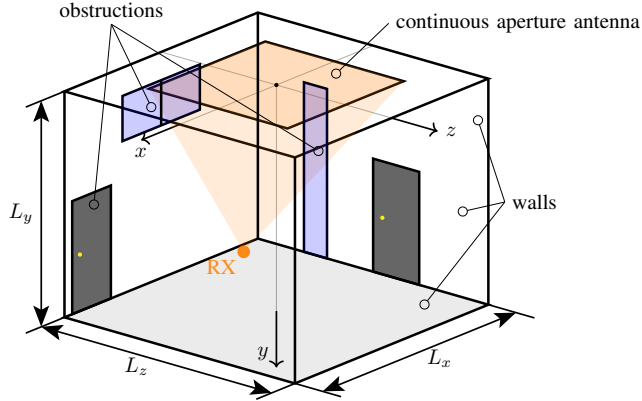

\subsection{ENVIRONMENT}\label{Subsection: Environment Model}
The environment considered in this paper is a three-dimensional cuboid room, where $L_x$, $L_y$, and $L_z$ define the size of the environment (see Fig. \ref{fig: System Model Overview}). 
The origin of the coordinate system is located in the middle of the ceiling of the room.
{\color{black} The receivers, e.g., sensors or wearable IoT devices, which may be subject to stringent design requirements, have to be supplied with power wirelessly.
%and are considered to lie in the radiating near-field of the continuous transmit antenna array, which is defined in Section II-\ref{Subsection: Transmit Antenna Model}.
%We emphasise that certain types of receivers, e.g., wearable IoT devices, are subject to stringent design requirements regarding their weight, size, and aesthetics.
}
{\color{black} To this end, the receivers are considered to be equipped with an antenna for harvesting energy.
%, such as, the mmWave rectenna discussed in \cite{Collado13} with an area of approximately $\lambda/2 \cdot \lambda/2$, where $\lambda$ is the wavelength.
%We define a receiver location as a small surface from which a receiver is capable of harvesting energy. The area of this surface is defined as $\lambda^2 /4$.
}
In this paper, we assume the transmit antenna {\color{black}architecture, which is realised via a continuous aperture antenna,} is located on the ceiling. 
The $y$-component of the room (see Fig. \ref{fig: System Model Overview}) is defined as the following closed interval in the set of real numbers $\mathbb{R}$
\begin{align}\label{eq: Receiver position in Y}
    {\color{black}\hat{\mathcal{Y}}} = \left\{ y \, \big\vert \, 0 \leq y \leq L_y \, , \; L_y > 0 \right\}.
\end{align}
The transmit antenna {\color{black}architecture} is located at $y=0$ and is defined in Section II-\ref{Subsection: Transmit Antenna Model}.
{\color{black}
At any other height level of the room, i.e., $y>0$, there exists a plane ${\color{black}\hat{\mathcal{X}} \times \hat{\mathcal{Z}}}$.
%Within each area, the EM field is assumed to be approximately constant.
%Hereby, each area must be large enough to accommodate the EH antenna of a receiver, which we assume to be $\lambda^2 / 4$.}
%At any other height level of the room, i.e., $y>0$, there exists a plane ${\color{black}\hat{\mathcal{X}} \times \hat{\mathcal{Z}}}$ containing areas from which receivers can harvest energy.
%Specifically, a receiver location is defined as the triplet $(x,y,z) \in {\color{black}\hat{\mathcal{X}} \times \hat{\mathcal{Y}}\times \hat{\mathcal{Z}}}$ which corresponds to a surface element of size $S_R = \lambda^2/4$ centred around $(x,z)$ at height $y>0$.
}
%{\color{black}\hat{\mathcal{X}} \times \hat{\mathcal{Z}}}$ at height $y$
A receiver location is defined as the triplet $(x_i,y,z_i) \in {\color{black}\hat{\mathcal{X}} \times \hat{\mathcal{Y}}\times \hat{\mathcal{Z}}}$ which corresponds to the $i$-th surface element of ${\color{black}\hat{\mathcal{X}} \times \hat{\mathcal{Z}}}$ at height $y$.
{\color{black}The receiver's antenna, which has the finite aperture $S_{RX}$, is assumed to be centred at $(x_i,y,z_i)$ and to ensure analytical tractability, we assume that the antennas are centred over a lattice defined in $\hat{\mathcal{X}} \times \hat{\mathcal{Z}}$. 
Therefore, ${\color{black}\hat{\mathcal{X}} \times \hat{\mathcal{Z}}}$ is partitioned into disjoint areas of size $S_{RX}$ from which a receiver is capable of harvesting energy.} 
For better legibility, the index $i$ is dropped and a receiver location is written as $(x,y,z) \in {\color{black}\hat{\mathcal{X}} \times \hat{\mathcal{Y}}\times \hat{\mathcal{Z}}}$.
{\color{black}
As a consequence of the position of the origin, ${\color{black}\hat{\mathcal{X}}}$ and ${\color{black}\hat{\mathcal{Z}}}$ are bounded by $\pm L_x/2$ and $\pm L_z/2$, respectively.}

{\color{black}\emph{Remark 1:} We note that for the problem considered in Section \ref{section: Methods}, the receivers are assumed to lie beyond the reactive near-field of the transmitter, see also Section II-\ref{Subsection: Channel Model}.
%The set of receiver positions in the radiating near-field is denoted by ${\color{black} \mathcal{X} \times \mathcal{Y} \times \mathcal{Z}\subset\hat{\mathcal{X}} \times \hat{\mathcal{Y}} \times \hat{\mathcal{Z}}}$.
}

\subsection{TRANSMIT ANTENNA MODEL}\label{Subsection: Transmit Antenna Model}
{\color{black}In Section \ref{Section: Introduction}, we have established that high-frequency systems are attractive and affordable for short-range applications, such as indoor WPT \cite{Zhanga}.
To leverage this, we consider the WPT system to operate in the mmWave frequency band in this paper.}
Primarily, we consider a two-dimensional transmit antenna {\color{black}architecture}.
Additionally, we investigate the problem for a one-dimensional transmit antenna {\color{black}architecture} which may be preferable in case of obstructions in the area of deployment.

The two-dimensional transmit antenna {\color{black}architecture} is deployed on the ceiling, i.e., located in the plane $y=0$, and is defined as the continuous surface $\mathcal{X}_a \times \mathcal{Z}_a$ comprising {\color{black}an uncountable infinity of} infinitesimally small {\color{black}radiating elements} with
\begin{align}
    \mathcal{X}_a = \left\{ a_x \, \big\vert \, -L_{a_x}/2 \leq a_x \leq L_{a_x}/2 \, , \; 0 < L_{a_x} \leq L_x \right\},
\end{align}
and
\begin{align}
    \mathcal{Z}_a = \left\{ a_z \, \big\vert \, -L_{a_z}/2 \leq a_z \leq L_{a_z}/2 \, , \; 0 < L_{a_z} \leq L_z \right\}.
\end{align}
By assuming a continuous transmit antenna {\color{black}architecture}, no constraint on the minimum size of a radiating element is imposed. The location of an infinitesimally small {\color{black}radiating element} is defined by the pair {\color{black}$(a_x,a_z) \in \mathcal{X}_a \times \mathcal{Z}_a$}.
% \begin{align}\label{eq: TX antenna positions continuous domain}
    
% \end{align}
Moreover, we assume that the surface, where the antenna {\color{black}architecture} is located, i.e., $\mathcal{X}_a \times \mathcal{Z}_a$, is free of obstructions. 
{\color{black} We note that from a practical point of view, obstructions in the ceiling, such as lights, vents, or sound dampening structures, will prohibit the deployment of excessively large continuous surfaces in typical indoor environments, such as, an office or a room in a residential building.}
The radiating elements of the transmit antenna {\color{black}architecture} are modelled as identical, omnidirectional antennas.
The transmit antenna {\color{black}architecture} model $\mathcal{X}_a \times \mathcal{Z}_a$ captures any arrangement of radiating elements ranging from the deployment of a single transmit antenna up to the distribution of radiating elements over a continuous planar surface\footnote{Note that {\color{black}adopting} a discrete transmit antenna model would exclude a continuous planar surface as a {\color{black}possible} solution to Problem \eqref{Problem: General Problem} discussed in Section \ref{section: Methods}. This would preclude the investigation of whether continuous transmit antenna architectures are optimal for maximising the received signal power at the worst possible receiver positions in the environment or not.}.

As a special case of this transmit antenna {\color{black}architecture} model, we also consider a one-dimensional transmit antenna {\color{black}architecture} $\mathcal{L}_a$ with an arbitrary direction in the plane $y=0$, where the location of a {\color{black}radiating elements} is defined as {\color{black}$a_l \in \mathcal{L}_a$,}
% \begin{align}\label{eq: TX antenna positions continuous domain - 1D}
% \end{align}
and $\mathcal{L}_a$ is assumed to be continuous.

{\color{black}We note that we neglect the impact of mutual coupling among the transmit {\color{black}radiating} elements to ensure analytical tractability of the problem\footnote{{\color{black}An investigation quantifying the impact of mutual coupling and other reactive near-field effects on system design is an interesting topic for future work. We point out that, in Section IV-\ref{Subsection: Fundamental Analysis --> Sparsity of optimal distribution}, we show that the optimal deployment only requires a finite number of transmit antennas as illustrated in Fig. 4.
We also note that mutual coupling effects are typically neglected for discrete antenna arrays \cite[Footnote 4]{bjornson20}.}}.}

\emph{Remark 2:} In this paper, we limit our study to transmit antenna {\color{black}architectures} deployed on the ceiling of the room, i.e., at $y=0$.
Nevertheless, our results can be extended to the case where the transmit antenna {\color{black}architecture} is integrated as part of a wall by exchanging the role of the $y$-component with the $x$-component or the $z$-component.

\subsection{CHANNEL MODEL}\label{Subsection: Channel Model}
%In this paper, we consider two channel models, namely an LoS channel discussed in Section II-\ref{Subsection: Channel Model}-\ref{Subsection: LoS Channel}, and a Rician fading channel described in Section II-\ref{Subsection: Channel Model}-\ref{Subsection: Rician Channel}. 
%{\color{black}In Section II-\ref{Subsection: Transmit Antenna Model}, we defined the system to operate in mmWave frequency band.
{\color{black}In Section \ref{Section: Introduction}, we discussed that the wireless channel is predominantly LoS in the mmWave frequency band since multipath fading components are weak at such a high frequency.}
The problem and the proposed solution presented in Section \ref{section: Methods} are developed under LoS conditions.
%based on the {\color{black} LoS channel model discussed in this subsection.}
%in Section II-\ref{Subsection: Channel Model}-\ref{Subsection: LoS Channel}. 
{\color{black}Assuming an LoS channel model} allows establishing an intuitive connection between the maximisation of the received power and the geometry underlying the problem.
The robustness of the optimal solution is investigated through simulation in Section \ref{Section: Simulation Results}, by assuming isotropic scattering in the half-space in front of the transmit antenna {\color{black}architecture} in addition to a strong LoS component, based on the Rician fading channel model. 

%\subsubsection{LoS Channel}\label{Subsection: LoS Channel}
% The Fresnel region of a wireless system is characterised through the relationship between the wavelength of the EM waves, the size of the transmit antenna array, and the distance between the transmitter and the receiver. 
% Specifically, a receiver is located within the Fresnel region if the distance between the receiver and the transmitter lies between the Fresnel distance $d_\mathrm{Fresnel} = \sqrt[3]{L^4 / 8 \lambda}$ and the Fraunhofer distance $d_\mathrm{Fraunhofer} = 2 L^2 / \lambda$, with $L$ and $\lambda$ being the diameter of the antenna array’s aperture and the wavelength, respectively \cite{Zhang2021}. 
% Specifically, $L$ denotes the maximum distance between any two radiating elements in $\mathcal{X}_a \times \mathcal{Z}_a$ or in $\mathcal{L}_a$, respectively.
{\color{black} A suitable channel model must capture the propagation characteristics of the EM wavefronts from the transmitter to a receiver accurately.
    The propagation environment of EM wavefronts can be divided into the near-field and the far-field region.
    In conventional wireless systems, the near-field region is negligibly small, however, when utilising the mmWave frequency band, which allows the EH antenna of a device to be small, the size of the near-field region becomes significant. 
    The relationship between the wavelength of the EM waves, the size of the transmit antenna architecture, and the distance between the transmitter and the receiver determines the size of the operating regions \cite{Liu23}.
    Specifically, a receiver is located within the near-field region if the distance between the receiver and the transmitter is below the Fraunhofer distance $d_\mathrm{Fraunhofer} = 2 L^2 / \lambda$, where $L$ and $\lambda$ are the antenna's aperture and the wavelength, respectively \cite{Liu23}. 
    The near-field region is further subdivided into the reactive near-field and the radiating near-field region, where the latter is commonly referred to as the Fresnel region.
    %The reactive near-field region is limited to the space close to a transmit antenna, i.e., less than a few wavelengths, and within this region the EM field is primarily reactive \cite{Liu23}.
    In this paper, we only consider the region beyond the reactive near-field for deployment, i.e., we restrict the receivers not to be located in the reactive near-field of the transmit antenna architecture.
    %\footnote{{\color{black}In the investigated simulation scenarios in Table 1, the optimal deployment of antennas is a set of radiating elements with a large inter-element distance relative to their size such that the receiver in the environment are located in the Fresnel region. Reactive near-field effects occurring less than a few wavelengths around each transmit antenna are thus deemed negligible.}}
    The motivation behind this is twofold. 
    Firstly, the size of the reactive near-field region is inherently limited since the deployment of overly large surfaces seems unrealistic (see Section II-\ref{Subsection: Transmit Antenna Model}).
    Secondly, in practice, receivers are unlikely to be located close to the ceiling\footnote{{\color{black} 
    We point out that, in Section IV-\ref{Subsection: Fundamental Analysis --> Sparsity of optimal distribution}, we show that the optimal deployment only requires a finite number of transmit antennas, which is realised with a discrete antenna array, and we illustrate this in Fig. \ref{fig: MRT heatmaps}.
    %The largest aperture of a transmit antenna $L_\mathrm{element}$ we consider in Section IV-\ref{Subsection: Fundamental Analysis --> Sparsity of optimal distribution} is $L_\mathrm{element} \approx 0.12$ m. 
    %We note that reactive near-field effects are negligible beyond a distance of
    %$d_\mathrm{Fresnel} = \sqrt[3]{\frac{L_\mathrm{element}^4}{8 \lambda}}$ from an antenna \cite{LiuHanzo23, Liu23}. 
    %Moreover, in Section \ref{Section: Simulation Results}, we define the boundaries of the Fresnel region.
    }}.
    The set of receiver positions beyond the reactive near-field is denoted by 
    \begin{align}
        \mathcal{X} \times \mathcal{Y} \times \mathcal{Z}\subset\hat{\mathcal{X}} \times \hat{\mathcal{Y}} \times \hat{\mathcal{Z}},
    \end{align}
    and any receiver position $(x,y,z)$ considered in the following satisfies 
    \begin{align}
        (x,y,z) \in \mathcal{X} \times \mathcal{Y} \times \mathcal{Z}.
    \end{align}
    {Therefore, receiver position $(x,y,z)$ lies in the radiating near-field or the far-field region of the transmitter.}
    %In the radiating near-field region, the wireless channel is modelled using spherical EM wavefronts.
    
%The size of the near-field region grows with increasing carrier frequency.

For high carrier frequencies, such as in the mmWave band, the reflection and scattering losses are severe, causing the LoS component to be dominant in the wireless channel \cite{Zhang2022, Do21}.}
Thus, an LoS channel is typically assumed for wireless systems operating in the near-field region \cite{Li2022, Zhanga, Zhang2021}.
%{\color{black}The accuracy of the LoS channel model increases with increasing operating frequency.}
We note that interruptions of the LoS due to obstructions cause blockages \cite{Zhanga} and we reduce the risk of channel blockage by placing the transmit antenna {\color{black}architecture} on the ceiling of the room instead of on one of the walls (see Section II-\ref{Subsection: Transmit Antenna Model}).
For the analytical tractability of the system design, we make the assumption that an LoS connection between the {\color{black}radiating elements} and the receiver exists and do not consider the impact of channel blockages in this paper.
{\color{black} 
    Under free-space LoS conditions, the equivalent complex baseband channel between the radiating element at $(a_x,a_z)$ and a receiver at $(x,y,z) \in \mathcal{X} \times \mathcal{Y} \times \mathcal{Z}$ is given by \cite{Liu23}\footnote{{\color{black} Including effective aperture losses and polarisation losses and studying their impact on the system performance is an interesting topic for future work. Here, we aim at determining the maximum achievable system performance and focus on providing analytical insights, and thus, we neglect these effects.}}
    \begin{align}\label{eq: LOS Channel}
        g_{xyz}(a_x,a_z) = \frac{\sqrt{c}}{D_{xyz}(a_x,a_z)} \, \mathrm{e}^{-\mathrm{j}\frac{2\pi}{\lambda} D_{xyz}(a_x,a_z)},
    \end{align}
    where $D_{xyz}(a_x,a_z)$ is the distance between the radiating element at $(a_x,a_z)$ and a receiver at $(x,y,z)$, i.e., $D_{xyz}(a_x,a_z)=\sqrt{(x-a_{x})^2 + y^2 + (z-a_{z})^2}$ and $c$ is the channel power gain at the reference distance of 1 meter.}
By virtue of the spherical nature of the EM wavefronts, both the amplitude and the phase of $g_{xyz}(a_x,a_z)$ depend on distance $D_{xyz}(a_x,a_z)$. In contrast, when the EM wavefronts are modelled as planar waves (far-field approximation), the \emph{path loss}{\color{black}, which is given by the amplitude in \eqref{eq: LOS Channel},} between all transmit {\color{black}radiating} elements and a receiver is assumed to be constant.

%\subsubsection{Rician Fading Channel}\label{Subsection: Rician Channel}

\subsection{TRANSMIT SIGNAL MODEL}\label{subsection: Transmit Signal Model} 
{\color{black}We denote the transmit power distribution as $P$, which is related to the transmit power density $p$ as follows
\begin{align}\label{eq: distribution}
     \mathrm{d} P(a_x,a_z) = p(a_x,a_z){\mathrm{d}a_x \mathrm{d}a_z}.
\end{align}
The total transmit power is constrained to be finite, i.e.,
\begin{align}\label{eq: Finite Power}
    \iint_{\mathcal{X}_a \times \mathcal{Z}_a} p(a_x,a_z) \,\mathrm{d}a_x\,\mathrm{d}a_z =
    \iint_{\mathcal{X}_a \times \mathcal{Z}_a}  \,\mathrm{d}P(a_x,a_z) =
    1,
\end{align}
and, without loss of generality, the transmit power is normalised to unity.

%Moreover, the received signal power depends on the transmitted signal. 
The transmit signal $t(a_x,a_z)$ emitted by the transmit antenna located at $(a_x,a_z)$ is given by
\begin{align}\label{eq: TS1 TX Signal}
     t(a_x,a_z) = \sqrt{p(a_x,a_z)} \, \mathrm{e}^{\mathrm{j}\theta(a_x,a_z)},
\end{align}
where $p(a_x,a_z) \geq 0$ is the non-negative transmit power density at $(a_x,a_z)$ and $\theta(a_x,a_z)$ is the phase.
Since there is no information on the location of the receivers, the transmit signal should illuminate the entire environment. Therefore, the phases are assumed to be independent and identically distributed (i.i.d.) among the radiating elements.
%where $n$ is additive white Gaussian noise with zero mean and variance $\sigma_n^2$, i.e., $n \sim \mathcal{CN}(0,\sigma_n^2)$. 

The received signal in position $(x,y,z)$ is given by
\begin{align}\label{eq: TS1 RX Signal}
      r_{xyz} =\iint_{\mathcal{X}_a \times \mathcal{Z}_a}
      t(a_x,a_z)
      g_{xyz}(a_x,a_z)
    \,\mathrm{d}a_x\,\mathrm{d}a_z,
\end{align}
and the power of the signal collected by the receiver's antenna with aperture $S_{RX}$, which is centred at position $(x,y,z)$, is given by
\begin{align}\label{eq: TS1 RX Power}
     \gamma_{xyz}
     &= S_{RX}\, \mathrm{E}\left\{\left\vert r_{xyz} \right\vert^2 \right\} \nonumber,\\
     &\overset{\text{i.i.d.}}{=} S_{RX}\, \iint_{\mathcal{X}_a \times \mathcal{Z}_a}
     \left\vert t(a_x,a_z)g(a_x,a_z) \right\vert^2
     \,\mathrm{d}a_x\,\mathrm{d}a_z  \nonumber,\\
     &= c\,S_{RX}\, \iint_{\mathcal{X}_a \times \mathcal{Z}_a}
     \frac{p(a_x,a_z)}{[D_{xyz}(a_x,a_z)]^2}
     \,\mathrm{d}a_x\,\mathrm{d}a_z,
\end{align}
where $\mathrm{E}\{ \cdot \}$ denotes the expectation operator and we exploit that the phases $\theta(a_x,a_z)$ are i.i.d.
%In the second line of \eqref{eq: TS1 RX Power} we exploit that the phases are i.i.d..
Therefore, the power received by the receiver in position $(x,y,z)$ is directly proportional to 
\begin{align}\label{eq: TS1 RX Power Integral}
     \Phi_P\{ f_{xyz}(a_x,a_z) \} = \iint_{\mathcal{X}_a \times \mathcal{Z}_a} f_{xyz}(a_x,a_z) \,\mathrm{d}P(a_x,a_z), 
\end{align}
with
\begin{align}\label{eq: f_xyz}
     f_{xyz}(a_x,a_z) = \left( \frac{1}{D_{xyz}(a_x,a_z)} \right)^2.
\end{align}
Operator $\Phi_P\{ \cdot \}$ represents the integration of function $f_{xyz}(a_x,a_z)$ with respect to (w.r.t.) the transmit power distribution $P$.
Mathematically, this is equivalent to the expectation of a function $f_{xyz}(a_x,a_z)$ w.r.t. a probability distribution.
This equivalence is crucial and it will be exploited in Section III-\ref{Subsection: Optimal Structure of Power Distribution}, where we will leverage the mathematical equivalence between determining the optimal transmit power distribution across a continuous aperture antenna, which we study in this paper, and the problem of finding the capacity-achieving amplitude-constrained input distribution of a channel, which has been studied in, e.g., \cite{smith1971information,Morsi2020,Dytso2018}.
Note that the impact of additive noise on the received power is neglected in \eqref{eq: TS1 RX Power}, \eqref{eq: TS1 RX Power Integral}.
Furthermore, we note that the amount of power harvested by a device is a non-decreasing function of the received signal power \cite{Clerckx20}. }

\section{OPTIMAL TRANSMIT POWER DISTRIBUTION}\label{section: Methods}
{\color{black}Our goal is to supply power to an arbitrary number of receivers in the environment without any knowledge of their location.}
This objective can be achieved by maximising the received signal power for the worst possible receiver position(s) in the environment, which leads to a max-min optimisation problem.
{\color{black}In this section, we formulate this optimisation problem to optimise the transmit power distribution $P$.
Hereby, the optimal transmit power distribution $P^*$ characterises the optimal transmit antenna architecture.
Deploying an HMIMO system would be necessary if we found that the optimal transmit power distribution's support had the same dimensionality as the distribution’s domain.
%, otherwise it is optimal to deploy a transmit antenna at every mass point $(a_x,a_z)$ where the optimal transmit power density $p^*$ is non-zero.
However, we will prove in Section III-\ref{Subsection: Optimal Structure of Power Distribution} that the optimal transmit power distribution's support has a lower dimensionality than the distribution's domain and thus, a HMIMO system is unnecessary in the context of the considered WPT problem.} 

\subsection{PROBLEM FORMULATION}\label{Subsection: Problem Formulation}
{\color{black} Maximising the power at the receiver's antenna centred at point $(x,y,z) \in \mathcal{X} \times \mathcal{Y} \times \mathcal{Z}$ in the environment is achieved by maximising $\Phi_{P}\{ f_{xyz} \}$, defined in \eqref{eq: TS1 RX Power Integral}, via an optimal distribution of transmit power among the radiating elements, where $P$ is a member of the set of distributions $\Omega$, i.e., $P(a_x,a_z) \geq 0, \forall (a_x,a_z) \in \mathcal{X}_a \times \mathcal{Z}_a$, and \eqref{eq: Finite Power} is satisfied.}
Hereby, the non-negativity and finite power constraints, which have been introduced in Section II-\ref{Subsection: Transmit Antenna Model}, apply to every element of $\Omega$.
Moreover, the function $P(a_x,a_z)$ is restricted to the set $\mathcal{X}_a \times \mathcal{Z}_a$, and as a consequence of \eqref{eq: Finite Power}, $P(a_x,a_z)$ is bounded by one and thus, {\color{black}is} square-integrable. 
Therefore, the set of considered distributions $\Omega$ is a subset of the square-integrable (finite energy) functions defined over the transmit antenna domain and denoted as $\mathcal{L}^2(\mathcal{X}_a \times \mathcal{Z}_a)$, i.e., $P(a_x,a_z) \in \Omega \subset \mathcal{L}^2(\mathcal{X}_a \times \mathcal{Z}_a)$. 
The set $\Omega$ is further characterised in the following lemma.
\begin{lemma}\label{lemma: Omega set}
    The set of distributions $\Omega \subset \mathcal{L}^2(\mathcal{X}_a \times \mathcal{Z}_a)$ is a weakly compact space in the weak topology and convex.
\end{lemma}
\hspace*{-3mm}\textit{Proof:} The proof is found in Appendix \ref{app: omega set}.

{\color{black} Ensuring that receivers are reliably supplied with power is achieved by maximising the received signal power at the worst possible receiver position(s). }
This can be formulated as the following max-min optimisation problem
\begin{align}\label{Problem: General Problem}
    \underset{P \in \Omega}{\text{maximise}}\quad \!\min_{\substack{(x,y,z)\\ \in \mathcal{X} \times \mathcal{Y}\times \mathcal{Z}}} \,  \Phi_{P}\{ f_{xyz} \}.
\end{align} 
The best worst-case performance, i.e., the optimal objective value of Problem \eqref{Problem: General Problem}, is denoted by $m$ and is defined as 
\begin{align}\label{Problem: Optimiser P}
    m = \min_{\substack{(x,y,z)\\ \in \mathcal{X} \times \mathcal{Y}\times \mathcal{Z}}} \, \Phi_{P^*}\{ f_{xyz} \},
\end{align}
where $P^*$ is the optimal distribution of power.
Note that Problem \eqref{Problem: General Problem} is linear in $P$. The feasibility of Problem \eqref{Problem: General Problem} is ensured through the compactness of $\Omega$ shown in Lemma \ref{lemma: Omega set}, i.e., $\Omega$ is closed and bounded.

\subsection{CRITICAL RECEIVER POSITIONS}\label{Subsection: Critical Receiver Positions}
By identifying the receiver locations where the received signal power is the lowest, the complexity of solving Problem \eqref{Problem: General Problem} can be reduced. 
These critical locations are defined w.r.t. the optimal distribution of power $P^*$.
To this end, the set containing the critical receiver positions when power is allocated optimally is denoted by $\mathcal{R_{\mathrm{crit}}}$.
The set of non-critical positions is the complement of $\mathcal{R_{\mathrm{crit}}}$, i.e., $\mathcal{R_{\mathrm{crit}}^\text{c}}$.
{\color{black}After the power has been optimally allocated across the continuous transmit antenna surface, the signal power received by an antenna centred at a critical receiver position $(x,y,z) \in \mathcal{R_{\mathrm{crit}}}$ will be below the signal power received by an antenna centered at any non-critical location $(x,y,z) \in \mathcal{R_{\mathrm{crit}}^\text{c}}$.}
Next, the set of critical receiver positions is characterised more precisely.
{\color{black}\begin{observation}\label{lemma: Restriction to necessary Locations}
    For any $(x,z) \in \mathcal{X} \times \mathcal{Z}$, the worst performance is attained at $y=L_y$. Consequently, all critical positions are located at $y=L_y$, i.e., $\mathcal{Y_{\mathrm{crit}}}=\{L_y \}$.
\end{observation}}
{\color{black}Key Observation \ref{lemma: Restriction to necessary Locations} holds since $f_{xyz}$ is a monotonically decreasing function of $y$.
Thus, the lowest objective value is obtained in the plane with $y=L_y$.
As a consequence of Key Observation \ref{lemma: Restriction to necessary Locations}\footnote{\color{black} Note that the result in Key Observation \ref{lemma: Restriction to necessary Locations}, i.e., $\mathcal{Y_{\mathrm{crit}}}=\{L_y \}$, holds since the radiating elements are assumed to be omnidirectional antennas and dropping this assumption may change the result.}, the function $f_{xyz}$ in the objective function of Problem \eqref{Problem: General Problem} can be equivalently substituted by $f_{xz}(a_x,a_z) := f_{xyz}(a_x,a_z)|_{y=L_y}$. 
%Note that $f_{xz}(a_x,a_z)$ is a bounded and square-integrable function.
}
%Thus, 
Problem \eqref{Problem: General Problem} is equivalently reformulated in the following lemma.  
\begin{lemma}\label{lemma: Lagrangian General Problem}
    Assume that the set of critical positions $\mathcal{R_{\mathrm{crit}}} \subseteq \mathcal{X}\times\mathcal{Y}_{crit}\times\mathcal{Z}$ is known. Then, solving Problem \eqref{Problem: General Problem} is equivalent to solving the following problem
    \begin{align}\label{Problem: Lagrangian General Problem}
        \underset{o \in \mathbb{R}, P \in \Omega, \Lambda}{\text{maximise}}\quad \!
        \underbrace{o - \sum_{\mathcal{R_{\mathrm{crit}}}} \left( o - \Phi_{P}\{ f_{xz} \} \right) \lambda_{xz}}_{J(P)},
    \end{align}
    %{\color{black}\begin{align}\label{Problem: Lagrangian General Problem}
    %     \underset{o \in \mathbb{R}, P \in \Omega, \Lambda}{\text{maximise}}\quad \!
    %     \underbrace{o - \iint_{\mathcal{R_{\mathrm{crit}}}} \left( o - \mathrm{E}_{P}\{ f_{xz} \} \right) \lambda_{xz} \, \mathrm{d}x\,\mathrm{d}z}_{J(P)},
    % \end{align}}
    where $\Lambda$ is the set of positive variables $\lambda_{xz} > 0, \forall (x,y,z) \in \mathcal{R_{\mathrm{crit}}}$.
    When allocating power optimally through $P^*(a_x,a_z)$, the optimal objective value is $o^* = m$ and the following conditions should also hold 
    \begin{align}\label{eq: Complementary Slackness}
       [&(\lambda_{xz} > 0 \implies m = \Phi_{P^*}\{ f_{xz} \})  \,\,\, \mathrm{or} \,\,\, \nonumber \\ 
       &(m < \Phi_{P^*}\{ f_{xz} \} \implies \lambda_{xz} = 0)], \nonumber \\  &\forall (x,y,z) \in \mathcal{X}\times\mathcal{Y_{\mathrm{crit}}}\times\mathcal{Z}.
    \end{align}
\end{lemma}
\hspace*{-3mm}\textit{Proof:}
    The proof is presented in Appendix \ref{app: Lagrangian}.
    
The objective function in Problem \eqref{Problem: Lagrangian General Problem} is referred to as $J(P)$ in the following.
Note that Problem {\color{black}\eqref{Problem: General Problem}} is a convex optimisation problem, which is readily apparent when the problem is stated in its epigraph form \cite{boyd2004convex} (see Appendix \ref{app: Lagrangian}).

\subsection{OPTIMAL STRUCTURE OF POWER DISTRIBUTION}\label{Subsection: Optimal Structure of Power Distribution}
In this section, we establish a relation between Problem \eqref{Problem: Lagrangian General Problem} and the problem of finding the optimal amplitude-constrained input distribution for reaching the information capacity of Gaussian channels studied in \cite{smith1971information,Morsi2020,Dytso2018}. 
In \cite{smith1971information} and \cite{Morsi2020}, the authors determine the optimal input distribution of transmit symbols, which are bounded in amplitude, for achieving information capacity, whereas our problem lies in identifying the optimal {\color{black} distribution of transmit power}, under geometrical constraints regarding the {\color{black} continuous transmit antenna architecture}. 
Beyond communication problems \cite{smith1971information}, similar results on the optimal input distribution in WPT applications have been established in \cite{Morsi2020}. 
First, we consider the optimisation of a two-dimensional distribution $P(a_x,a_z)$ which is significantly different from the problems involving one-dimensional distributions such as \cite{smith1971information} and \cite{Morsi2020}. Specifically, the properties of the optimal {\color{black}two-dimensional} distribution's structure are less stringent.
The analysis of optimisation problems involving input distributions with a dimensionality greater than one has been studied in \cite{Dytso2018}. 
Subsequently, we consider the optimisation problem with a one-dimensional distribution, which allows characterising the optimal distribution as in \cite{smith1971information} and \cite{Morsi2020}. 

In the following, we investigate properties, which apply to the optimal distribution of power and allow the characterisation of the optimal {\color{black}transmit antenna architecture}. The main results are stated in Propositions \ref{proposition: Nowhere dense set} and \ref{proposition: Measure zero} for the optimal distribution of power over a two-dimensional surface and in Proposition \ref{proposition: 1D result} for the optimal one-dimensional distribution.
First, a necessary and sufficient condition for the optimal distribution of power $P^*$ is established in Proposition \ref{proposition: Necessary and Sufficient Constraints} and Corollary \ref{corollary: Optimality of distribution iff two equations hold}.

\begin{proposition}\label{proposition: Necessary and Sufficient Constraints}
    Consider the optimisation problem in \eqref{Problem: Lagrangian General Problem}.
    There exists an optimal distribution of power $P^*(a_x,a_z)$ which attains the maximum $m$ of Problem \eqref{Problem: Lagrangian General Problem} if and only if the following condition holds 
    \begin{align}\label{eq: nec und suf}
        \iint_{\mathcal{X}_a \times \mathcal{Z}_a} \, 
        \Bar{f}(a_x,a_z) \, \,\mathrm{d}P(a_x,a_z) \leq m, \quad \forall P \in \Omega,
    \end{align}
    where $\Bar{f}(a_x,a_z)$ is given by
    \begin{align}\label{eq: bar f}
        \Bar{f}(a_x,a_z) = \frac{
            \sum_{\mathcal{R}^{\mathrm{crit}}} \, 
            f_{xz}(a_x,a_z) \lambda_{xz} 
        }
        {
            \sum_{\mathcal{R}^{\mathrm{crit}}} \, 
            \lambda_{xz} 
        },
    \end{align}
    % {\color{black}\begin{align}\label{eq: bar f}
    %     \Bar{f}(a_x,a_z) = \frac{
    %         \iint_{\mathcal{R}^{\mathrm{crit}}} \, 
    %         f_{xz}(a_x,a_z) \lambda_{xz} \, \mathrm{d}x\,\mathrm{d}z
    %     }
    %     {
    %         \iint_{\mathcal{R}^{\mathrm{crit}}} \, 
    %         \lambda_{xz} \, \mathrm{d}x\,\mathrm{d}z
    %     },
    % \end{align}}
    and \eqref{eq: nec und suf} is satisfied with equality for the optimal distribution of power $P^*(a_x,a_z)$. 
\end{proposition}
\hspace*{-3mm}\textit{Proof:}    This proposition is proven in Appendix \ref{app: Necessary and Sufficient Constraints}.

Intuitively, function $\Bar{f}$ can be interpreted as the weighted average of $f_{xz}$ over the receiver positions $\mathcal{X}\times\mathcal{Z}$ at $y=L_y$. The weights are given by the variables $\lambda_{xz} / \sum_{\mathcal{R}^{\mathrm{crit}}} \, \lambda_{xz}$ and thus, only critical receiver positions contribute to $\Bar{f}$.
Note that the optimal distribution $P^*$ is not necessarily unique since $J(P)$ is linear in $P$ and not strictly concave as in the problems considered in \cite{smith1971information, Morsi2020}.
Additionally, Proposition \ref{proposition: Necessary and Sufficient Constraints} provides a condition for the existence of $P^*$. Obtaining $P^*$ requires determining the variables $\lambda_{xz}$ satisfying \eqref{eq: nec und suf} and \eqref{eq: bar f}.

Next, the optimality condition for $P^*$, given in Proposition \ref{proposition: Necessary and Sufficient Constraints}, is restated in a more intuitive manner. 
This requires defining the support of $P^*$, which is analogous to the definition in \cite{Dytso2018}.
A {\color{black}radiating element's position} $(a_x,a_z)$ is called a \emph{point of increase} of the optimal distribution $P^*$ if for any open subset $O \subset \mathbb{R}^2$ containing $(a_x,a_z)$, $P^*(O)>0$\footnote{Here, the argument of the optimal transmit power distribution is a set instead of a tuple of points, i.e., $P^*(O)$ corresponds to  $\iint_{O}\mathrm{d}P^*(a_x, a_z)$.} holds \cite{Dytso2018}.
The set of points of increase of $P^*$ is defined as $\mathcal{E}(P^*) \subseteq \mathbb{R}^2$. 
Note that $P^*(\mathcal{E}(P^*)) = 1$ \cite{Dytso2018}. 
Moreover, $\mathcal{E}(P^*)$ is the smallest closed subset of $\mathbb{R}^2$ such that $P^*(\mathcal{E}(P^*)) = 1$ \cite{Dytso2018}.
Note that $P \in \Omega$ is restricted to $\mathcal{X}_a \times \mathcal{Z}_a \subset \mathbb{R}^2$ and thus, $\mathcal{E}(P^*) \subseteq \mathcal{X}_a \times \mathcal{Z}_a$.
We define $\mathcal{E}(P^*)$ as the support of $P^*$.

\begin{corollary}\label{corollary: Optimality of distribution iff two equations hold}
    Consider the optimisation problem in \eqref{Problem: Lagrangian General Problem}. There exists an optimal distribution of power $P^*$ for Problem \eqref{Problem: Lagrangian General Problem} if and only if $\Bar{f}(a_x,a_z)$ satisfies the following equations 
    \begin{align}\label{eq: optimality condition 1}
         \Bar{f}(a_x,a_z) &\leq m, \quad \forall (a_x,a_z) \in \mathcal{X}_a \times \mathcal{Z}_a, \\
         \Bar{f}(a_x,a_z) &= m, \quad \forall (a_x,a_z) \in \mathcal{E}(P^*). \label{eq: optimality condition 2}
    \end{align}
\end{corollary}
\hspace*{-3mm}\textit{Proof:}
    The proof is provided in Appendix \ref{app: Optimality of distribution iff two equations hold}.

Thus, Corollary \ref{corollary: Optimality of distribution iff two equations hold} describes the relationship between the optimal distribution of transmit power $P^*$ for attaining the optimal objective value $m$ of Problem \eqref{Problem: Lagrangian General Problem} and the optimal set of {\color{black}positions regarding the radiating elements} $\mathcal{E}(P^*)$, i.e., the support of $P^*$.

Based on the optimality conditions on $P^*$ defined in Proposition \ref{proposition: Necessary and Sufficient Constraints} and Corollary \ref{corollary: Optimality of distribution iff two equations hold}, the set of optimal {\color{black}positions regarding the radiating elements} is characterised more precisely. 
First, we complete our analysis of the two-dimensional distribution $P^*(a_x,a_z)$. 
Subsequently, we characterise the set of optimal positions under the constraint of a {\color{black} continuous} one-dimensional deployment.
{\color{black} The characterisation requires the following definition in the two-dimensional scenario.}
{\color{black}\begin{definition}[see \cite{Dytso2018}]
A subset $S_1 \subset S_0$ is \emph{dense} in $S_0$ if every element $s \in S_0$ is either an element of $S_1$ or an accumulation point of $S_1$.
A subset $S_1 \subset S_0$ is \emph{nowhere dense} if $S_2 \cap S_1$ is not \emph{dense} in $S_0$, where $S_2$ is any non-empty open set $S_2 \subset S_0$.
\end{definition}}

\begin{proposition}\label{proposition: Nowhere dense set}
    The support $\mathcal{E}(P^*) \subset \mathcal{X}_a\times\mathcal{Z}_a$ of the optimal two-dimensional distribution of power $P^*$ is a nowhere dense set of $\mathcal{X}_a\times\mathcal{Z}_a$.
\end{proposition}
\hspace*{-3mm}\textit{Proof:}
    The proof is provided in Appendix \ref{app: Nowhere dense set}.

\begin{proposition}\label{proposition: Measure zero}
    The support $\mathcal{E}(P^*) \subset \mathcal{X}_a\times\mathcal{Z}_a$ of the optimal two-dimensional distribution of power $P^*$ is of Lebesgue measure zero.
\end{proposition}
\hspace*{-3mm}\textit{Proof:}
    The proof is provided in Appendix \ref{app: Measure zero}.

Thus, $\mathcal{E}(P^*)$ is a set of pairs $(a_x,a_z)$, which define the optimal {\color{black} positions of radiating elements}, that is nowhere dense and has Lebesgue measure zero, i.e., the set $\mathcal{E}(P^*)$ does not occupy any measurable surface in $\mathcal{X}_a \times \mathcal{Z}_a$.
Intuitively, the set $\mathcal{E}(P^*)$ could consist of isolated points scattered throughout $\mathcal{X}_a \times \mathcal{Z}_a$ since both properties hold in this case. 
On the other hand, Propositions \ref{proposition: Nowhere dense set} and \ref{proposition: Measure zero} do not preclude the possibility of $\mathcal{E}(P^*)$ containing more complex structures, such as continuous one-dimensional structures, as long as $\mathcal{E}(P^*)$ is nowhere dense and of measure zero. 
This is further investigated in Section \ref{Section: System Discretisation}.
{\color{black} We note that based on Propositions \ref{proposition: Nowhere dense set} and \ref{proposition: Measure zero}, a continuous planar transmit antenna architecture, which is utilised in HMIMO, is unnecessary for WPT.}
Moreover, observe that, by virtue of the max-min problem, the optimal distribution must be symmetrical w.r.t. the $x$-axis and the $z$-axis, i.e., $P^*(a_x,a_z) = P^*(-a_x,a_z) =P^*(-a_x,-a_z) =P^*(a_x,-a_z)$. Violating this symmetry condition would lead to a decrease in performance at the worst possible receiver positions, i.e., a decrease of the objective value.

The results presented in this section, which allowed drawing conclusions regarding the optimal deployment of transmit {\color{black}radiating elements} $\mathcal{E}(P^*)$, pertained to the analysis of a two-dimensional distribution of power. 
In the following, we restrict the {\color{black}continuous transmit antenna architecture} to lie on the one-dimensional set $\mathcal{L}_a$, which results in a problem analogous to \cite{smith1971information,Morsi2020} (see Section \ref{Subsection: Transmit Antenna Model}). Consequently, the optimal distribution of power $\Tilde{P}^*(a_l)$ along $\mathcal{L}_a$ is one-dimensional. This problem is investigated in the following, whereby the {\color{black}optimal set of positions of the radiating elements} $\Tilde{\mathcal{E}}(\Tilde{P}^*) \subset \mathcal{L}_a$ is shown to be finite.

\begin{proposition}\label{proposition: 1D result}
    The support $\Tilde{\mathcal{E}}(\Tilde{P}^*) \subset \mathcal{L}_a$ of the optimal one-dimensional distribution is a finite set of points.
\end{proposition}
\hspace*{-3mm}\textit{Proof:}
    The proof is provided in Appendix \ref{app: 1D result}.

We note that a continuous {\color{black}transmit antenna architecture} is unnecessary for one-dimensional transmit antenna architectures since Proposition \ref{proposition: 1D result} guarantees that a finite number of radiating elements is sufficient.

\section{SYSTEM DISCRETISATION}\label{Section: System Discretisation}
In the following, to shed light on our analytical results in Section \ref{section: Methods}, the optimal {\color{black} transmit antenna architecture} is studied via numerical analysis.
In order to solve Problem \eqref{Problem: General Problem} numerically, the {\color{black} continuous sets} $\mathcal{X}_a\times\mathcal{Z}_a$ and $\mathcal{L}_a$ have to be discretised.
%{\color{black}Having shown that a continuous aperture antenna is unnecessary, the optimal distribution of transmit power yields the optimal transmit antenna deployment as a byproduct since a transmit antenna should only be deployed if it is allocated non-zero transmit power.}
For a one-dimensional transmit antenna {\color{black}architecture}, we have proven in Proposition \ref{proposition: 1D result} that a finite number of transmit antennas placed in $\Tilde{\mathcal{E}}(\Tilde{P}^*) \subset \mathcal{L}_a$ is optimal.
For a two-dimensional transmit antenna {\color{black}architecture}, Proposition \ref{proposition: Measure zero} guarantees that the optimal deployment of transmit antennas is on a set $\mathcal{E}(P^*) \subset \mathcal{X}_a\times\mathcal{Z}_a$ with Lebesgue measure zero. 
We {\color{black}obtain approximations of} the optimal solution to Problem \eqref{Problem: General Problem} through discretising the problem, and subsequently, solving it numerically. 
{\color{black}We verify} numerically {\color{black}that} the solution of the optimisation problem remains unchanged for finer levels of discretisation (LoD), i.e., for an increasing granularity of the sampling grid.
Further, since Propositions \ref{proposition: Nowhere dense set} and \ref{proposition: Measure zero} ensure that the optimal set of transmit antennas $\mathcal{E}(P^*) \subset \mathcal{X}_a\times\mathcal{Z}_a$ is nowhere dense and has Lebesgue measure zero, which are properties satisfied by a set of isolated points scattered within $\mathcal{X}_a\times\mathcal{Z}_a$, we conjecture that such a set {\color{black}of isolated points} is optimal. Thereby, to support this conjecture by simulation results, it is necessary to investigate whether the number of transmit antennas increases monotonically for finer LoD or if the number of transmit antennas remains constant. 
Intuitively, if the latter is true, there is an optimal, finite number of transmit antennas, which depends on the geometrical properties of the environment, i.e., $L_x$, $L_y$, and $L_z$, at least from a practical point of view. 
The three-dimensional environments, which are investigated in this paper, are listed in Table \ref{tab: Considered Environments}. 
The difference between the environments is the ratio of height to width since the width and depth of the room are chosen to be equal. The environments are denoted by their height-to-width ratio. 
\begin{table}[h]
\renewcommand{\arraystretch}{1.3}
\caption{Definition of the four environments analysed in Sections \ref{Section: System Discretisation} and \ref{Section: Simulation Results}.}
\label{tab: Considered Environments}
\centering
\begin{tabular}{|c|c|c|c|}
\hline
Height $L_y$ in [m] & Width $L_x$ in [m] & Depth $L_z$ in [m] & $\frac{\mathrm{Height}}{\mathrm{Width}}$\\ \hhline{|=|=|=|=|}
2 & 2 & 2 & 1/1\\ \hline
2 & 6 & 6 & 1/3\\ \hline
2 & 8 & 8 & 1/4\\ \hline
2 & 10 & 10 & 1/5\\ \hline
\end{tabular}
\end{table}
{\color{black} We note that the dimensions of the environments listed in Table 1 are realistic, e.g., for an office or a room in a residential building.
In environments spanning a significantly larger area, such as, a factory building, mmWave operating frequencies become unaffordable due to the large path loss, which would result in a very low received power when receivers are located far away from a transmitter, and thus, operation at a lower frequency would be preferred.
However, NLoS components become increasingly important for decreasing operating frequency, and thus, the wireless channel may not be accurately described by a deterministic LoS channel.
In this case, it is preferable to consider alternative wireless systems and models, such as those in \cite{Rosabal21,Jia21,Zhou23,Tabassum15}, for which optimal transmit antenna deployment and power allocation differ significantly in methodology and solutions.}

\subsection{DISCRETISED OPTIMISATION PROBLEM}
In the following, the set of transmit antenna positions, i.e., $\mathcal{X}_a\times\mathcal{Z}_a$ for the two-dimensional problem and $\mathcal{L}_a$ for the one-dimensional problem, is discretised. 
In this paper, we only investigate systems where the {\color{black}discretised transmit antenna architecture, which we also refer to as a (discrete) antenna array,} covers the entire ceiling{\color{black}\footnote{{\color{black} 
In Section II-\ref{Subsection: Transmit Antenna Model}, the deployment of excessively large continuous surfaces in typical indoor environments was deemed unsuitable due to obstructions in the ceiling.
This problem is substantially less severe when deploying discrete transmit antennas within $\mathcal{X}_a\times\mathcal{Z}_a$ or $\mathcal{L}_a$ as considered in this section.
%We further investigate this in Section \ref{Section: Simulation Results}.
%From a practical point of view, obstructions in the ceiling, such as lights, vents, or sound dampening structures, will prohibit the deployment of excessively large continuous surfaces in typical indoor environments, such as, an office or a room in a residential building.
%While and thus, the problem of interference between obstructions in the ceiling and transmit antennas is substantially less severe compared to the case with a continuous architecture since an individual antenna consumes almost negligible space compared to a continuous array.
%By virtue of our proposed optimal deployment in Fig. 4, we observe that a small number of radiating elements is optimal in the considered environments and thus, interference with obstructions will indeed not be an issue in the majority of scenarios. 
%Additionally, for the one-dimensional scenario, we presented a rigorous proof that deploying a finite number of transmit antennas $\mathcal{L}_a$ is optimal.}
}}}, i.e., $L_{a_x} = L_x$ and $L_{a_z} = L_z$. 
The one-dimensional antenna array is also placed on the ceiling along the $x$-axis, i.e., $\mathcal{L}_a = \mathcal{Z}_a$.
Moreover, we {\color{black}set} $L_x = L_z$.
The LoD is defined as the number $N+1$ of equally-sized intervals into which the sets $\mathcal{X}_a$, $\mathcal{Z}_a$ and $\mathcal{L}_a$ are partitioned, respectively.
Consequently, we have a grid of antenna positions with $(N+1)^2$ elements in $\mathcal{X}_a \times \mathcal{Z}_a$ and $N+1$ elements in $\mathcal{L}_a$, respectively. 
The position of transmit antenna $i$, which corresponds to one element of the grid, is given as {\color{black}$(a_{x},a_{z})_i$}, with $i=1,\dots,(N+1)^2$, in two dimensions and {\color{black}$(a_{l})_i$}, with $i=1,\dots,N+1$, in one dimension.
{\color{black}In the following, the coordinates of position $(a_{x},a_{z})_i$ are denoted by $a_{xi}$ and $a_{zi}$.
Similarly, the coordinate of position $(a_{l})_i$ is given by $a_{li}$.}
%{\color{black} The receiver positions are discretised with the same LoD as the transmit antenna positions.}
By utilising the definitions in Section II-\ref{subsection: Transmit Signal Model}, Problem \eqref{Problem: General Problem} is discretised as follows
\begin{subequations}\label{Problem: Epigraph General Problem Discrete}
    \begin{alignat}{2}
    &\underset{o \in \mathbb{R}, p(a_{xi},a_{zi})}{\text{maximise}}
    &\qquad& o \label{Objective: Epigraph General Problem Discrete}\\
    &\text{subject to} 
    && o \leq \sum_{(a_{xi},a_{zi})} f_{xz}(a_{xi},a_{zi}) p(a_{xi},a_{zi}), \quad \nonumber \\ &&&\forall (x,L_y,z) \in \mathcal{X}\times\mathcal{Y_{\mathrm{crit}}}\times\mathcal{Z}, \label{Constraint: Epigraph General Problem Discrete} \\
    &&& \sum_{(a_{xi},a_{zi})} p(a_{xi},a_{zi}) = 1, \label{Constraint: Epigraph General Problem Discrete, sum finite} \\
    &&& p(a_{xi},a_{zi}) \geq 0,
    \label{Constraint: Epigraph General Problem Discrete, non-neg}
    \end{alignat}
    \end{subequations}
where the problem has been rewritten in epigraph form \cite{boyd2004convex}. 
After discretisation, the convex optimisation problem in \eqref{Problem: Epigraph General Problem Discrete} can be solved efficiently, for example, by applying the interior-point method \cite{boyd2004convex}. The problem is solved using standard numerical tools for convex optimisation, such as CVXPY \cite{diamond2016cvxpy}, \cite{agrawal2018rewriting}.
For all results shown, any value below $10^{-6}$ in the optimal distribution of power was deemed to a {\color{black} representation and computation error in computer arithmetics} and was set to zero. 
In the worst case, values below $10^{-6}$ accounted for less than $10^{-5}\%$ of the total power distributed among the transmit antennas.

\subsection{LEVEL OF DISCRETISATION}\label{Subsection: Fundamental Analysis --> Level of Discretisation}
There is an inherent trade-off between the computational complexity of solving a problem and the accuracy of its solution. The necessary LoD is determined based on a desired accuracy and a criterion to assess whether the proposed solution is sufficiently accurate. To this end, Problem \eqref{Problem: Epigraph General Problem Discrete} was solved repeatedly for increasing granularity of the sampling grid.
Subsequently, the differences between the objective values of consecutive iterations were evaluated. 

In Fig. \ref{fig: Objective value difference vs. n}, the absolute value of the relative difference between the objective values, i.e., the absolute value of the difference between the current and previous objective value divided by the current objective value, was calculated for every $N$. 
As the LoD increases, the accuracy increases for all investigated environments since the objective value stabilises. 
{\color{black} From Fig. \ref{fig: Objective value difference vs. n}, at LoD $N+1=81$, the relative difference lies below $0.05\%$ and beyond that, the improvements were deemed negligible. 
Consequently, unless specified otherwise, all simulations in the following are performed at LoD $N+1=81$.}
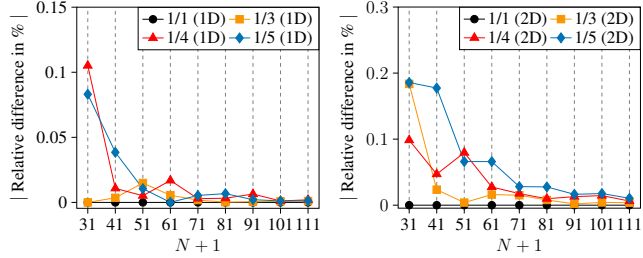
\begin{figure}[t]
    \centering
    \scalebox{0.47}{
    \begin{minipage}{0.5\textwidth}
        \centering
        % This file was created with tikzplotlib v0.10.1.
\begin{tikzpicture}

\definecolor{black}{RGB}{0,0,0}
\definecolor{orange}{RGB}{250,149,0}
\definecolor{sky-blue}{RGB}{86,180,233}
\definecolor{bluish-green}{RGB}{0,158,115}
\definecolor{yellow}{RGB}{240, 228, 66}
\definecolor{blue}{RGB}{0,114,178}
\definecolor{vermillion}{RGB}{255, 0, 0}
\definecolor{reddish-purple}{RGB}{204, 121, 167}

\begin{axis}[
name=ax1,
xmajorgrids=true,
grid style={help lines,dashed},
legend cell align={left},
legend columns=2, 
legend style={at={(1,1)}, draw opacity=1, text opacity=1, font=\Large},
tick align=outside,
tick pos=left,
xlabel={$N+1$},
xmin=27, xmax=115,
xtick style={color=black},
ylabel={$\vert$ Relative difference in \% $\vert$},
ymin=-0.004530827928549, ymax=0.15,
ytick={0,0.05,0.1,0.15},
yticklabels={0,0.05,0.1,0.15},
ytick style={color=black},
xtick={31,41,51,61,71,81,91,101,111},
xticklabels={31,41,51,61,71,81,91,101,111},
label style={font=\Large},
tick label style={font=\Large} 
]

\addplot [black, mark=*, thick, mark size=3pt]
table{%
x  y
31 1.00098040967112e-08
41 9.12148134801782e-10
51 3.80995235360615e-10
61 4.13336032067946e-11
71 4.79888462301403e-08
81 4.7987613882583e-08
91 6.94444501903035e-11
101 2.26478835685384e-09
111 2.33264518811893e-09
};
\addlegendentry{1/1 (1D)}

\addplot [orange, mark=square*, thick, mark size=3pt]
table{%
x  y
31 1.56725643574163e-07
41 0.00356341542218086
51 0.0148954767263176
61 0.00549884555144375
71 0.00178649181888924
81 0.000180216414247258
91 0.000541436571566134
101 0.000860405506952766
111 0.000984962168770753
};
\addlegendentry{1/3 (1D)}

\addplot [red, mark=triangle*, thick, mark size=4pt]
table{%
x  y
31 0.105092153601949
41 0.0107931752144408
51 0.00514989672097954
61 0.0168683544670145
71 0.00292387987358422
81 0.00319331045205429
91 0.00644662420428155
101 0.00113757730164643
111 0.00201142315092717
};
\addlegendentry{1/4 (1D)}

\addplot [blue, mark=diamond*, thick, mark size=4pt]
table{%
x  y
31 0.0830752978476457
41 0.0384968117405293
51 0.0105431050199001
61 6.21156481361851e-05
71 0.00537533646726152
81 0.00674493174241952
91 0.00210352862385355
101 0.00112420767468535
111 0.00128375973257322
};
\addlegendentry{1/5 (1D)}

\end{axis}

\end{tikzpicture}%
    \end{minipage}%
    }
    \hfill
    \scalebox{0.47}{
    \begin{minipage}{0.5\textwidth}
        \centering
        % This file was created with tikzplotlib v0.10.1.
\begin{tikzpicture}

\definecolor{black}{RGB}{0,0,0}
\definecolor{orange}{RGB}{250,149,0}
\definecolor{sky-blue}{RGB}{86,180,233}
\definecolor{bluish-green}{RGB}{0,158,115}
\definecolor{yellow}{RGB}{240, 228, 66}
\definecolor{blue}{RGB}{0,114,178}
\definecolor{vermillion}{RGB}{213, 94, 0}
\definecolor{reddish-purple}{RGB}{204, 121, 167}

\begin{axis}[
name=ax1,
xmajorgrids=true,
grid style={help lines,dashed},
legend cell align={left},
legend columns=2, 
legend style={at={(1,1)}, draw opacity=1, text opacity=1, font=\Large},
tick align=outside,
tick pos=left,
xlabel={$N+1$},
xmin=27, xmax=115,
xtick style={color=black},
ylabel={$\vert$ Relative difference in \% $\vert$},
ymin=-0.004530827928549, ymax=0.3,
ytick style={color=black},
xtick={31,41,51,61,71,81,91,101,111},
xticklabels={31,41,51,61,71,81,91,101,111},
label style={font=\Large},
tick label style={font=\Large}
]

%%% TRUE %%%
% \addplot [black, mark=*,dash dot]
% table{%
% x  y
% 31 1.92449722824506e-08
% 41 6.63363808328654e-09
% 51 -2.74333000760407e-08
% 61 2.58820631593437e-08
% 71 4.413136522885e-11
% 81 -7.11244396711663e-09
% 91 -2.79548273418584e-08
% 101 -4.670571707166e-08
% 111 8.19257106599025e-08
% };
% \addlegendentry{1/1 (2D)}

% \addplot [orange, mark=square*,dash dot]
% table{%
% x  y
% 31 0.183694665097456
% 41 0.0235044037345977
% 51 0.00390359234628823
% 61 -0.0160326992138748
% 71 0.0150854421137048
% 81 0.00788773319545522
% 91 -0.00198795321828538
% 101 -0.00415638096633586
% 111 0.00378371078727513
% };
% \addlegendentry{1/3 (2D)}
% \addplot [sky-blue, mark=triangle*, mark options={fill=sky-blue}, dash dot]
% table{%
% x  y
% 31 0.098477766431726
% 41 -0.0469765648228537
% 51 0.0795287975037584
% 61 -0.0275416581536403
% 71 0.017501197955927
% 81 0.00980427260733885
% 91 -0.0129068734782467
% 101 0.0145231938039059
% 111 -0.00620101264323658
% };
% \addlegendentry{1/4 (2D)}
% \addplot [bluish-green, mark=diamond*, dash dot]
% table{%
% x  y
% 31 -0.185739050754796
% 41 0.17734476074025
% 51 -0.0660816611204873
% 61 0.0659818850371807
% 71 -0.028140687084488
% 81 0.0276514806887507
% 91 -0.0164350057492046
% 101 0.0176245703772393
% 111 -0.010460989092187
% };
% \addlegendentry{1/5 (2D)}

%%% ABSOLUTE VALUE %%%
\addplot [black, mark=*, thick, mark size=3pt]
table{%
x  y
31 1.92449722824506e-08
41 6.63363808328654e-09
51 2.74333000760407e-08
61 2.58820631593437e-08
71 4.413136522885e-11
81 7.11244396711663e-09
91 2.79548273418584e-08
101 4.670571707166e-08
111 8.19257106599025e-08
};
\addlegendentry{1/1 (2D)}

\addplot [orange, mark=square*,thick, mark size=3pt]
table{%
x  y
31 0.183694665097456
41 0.0235044037345977
51 0.00390359234628823
61 0.0160326992138748
71 0.0150854421137048
81 0.00788773319545522
91 0.00198795321828538
101 0.00415638096633586
111 0.00378371078727513
};
\addlegendentry{1/3 (2D)}
\addplot [red, mark=triangle*, thick, mark size=4pt]
table{%
x  y
31 0.098477766431726
41 0.0469765648228537
51 0.0795287975037584
61 0.0275416581536403
71 0.017501197955927
81 0.00980427260733885
91 0.0129068734782467
101 0.0145231938039059
111 0.00620101264323658
};
\addlegendentry{1/4 (2D)}
\addplot [blue, mark=diamond*, thick, mark size=4pt]
table{%
x  y
31 0.185739050754796
41 0.17734476074025
51 0.0660816611204873
61 0.0659818850371807
71 0.028140687084488
81 0.0276514806887507
91 0.0164350057492046
101 0.0176245703772393
111 0.010460989092187
};
\addlegendentry{1/5 (2D)}

\end{axis}

\end{tikzpicture}%
    \end{minipage}
    }
    \caption{Absolute value of relative difference between objective values in percent for increasing $N$ for the height-to-width ratios defined in Table \ref{tab: Considered Environments}. 
    \label{fig: Objective value difference vs. n}
    }
\end{figure}

\subsection{FINITENESS OF OPTIMAL POWER DISTRIBUTION}\label{Subsection: Fundamental Analysis --> Sparsity of optimal distribution}
In this section, the optimal number of transmit antennas\footnote{The optimal number of transmit antennas is defined as the number of non-zero values in the solution $p^*(a_{ix},a_{iz})$ of Problem \eqref{Problem: Epigraph General Problem Discrete}.} $N_t^*$ is investigated for increasing LoD.
{\color{black}To this end, the optimal number of transmit antennas for each considered environment is displayed in Fig. \ref{fig: number antennas}. }
\begin{figure}[t]
    \centering
    \scalebox{0.47}{
    \begin{minipage}{0.5\textwidth}
        \centering
        % This file was created with tikzplotlib v0.10.1.
\begin{tikzpicture}

\definecolor{black}{RGB}{0,0,0}
\definecolor{orange}{RGB}{250,149,0}
\definecolor{sky-blue}{RGB}{86,180,233}
\definecolor{bluish-green}{RGB}{0,158,115}
\definecolor{yellow}{RGB}{240, 228, 66}
\definecolor{blue}{RGB}{0,114,178}
\definecolor{vermillion}{RGB}{213, 94, 0}
\definecolor{reddish-purple}{RGB}{204, 121, 167}

\begin{axis}[
name=ax1,
xmajorgrids=true,
grid style={help lines,dashed},
legend cell align={left},
legend columns=2,
legend style={at={(1,1)}, draw opacity=1, text opacity=1, font=\Large},
tick align=outside,
tick pos=left,
xlabel={$N+1$},
xmin=17, xmax=151,
xtick style={color=black},
ylabel={$N_t^*$},
ymin=-1.24530827928549, ymax=10,
ytick style={color=black},
xtick={21,41,61,81,101,121,141},
xticklabels={21,41,61,81,101,121,141},
label style={font=\Large},
tick label style={font=\Large}  
]

\addplot [black, mark=*,thick, mark size=3pt]
table{%
x  y
21 1
%31 1
41 1
%51 1
61 1
%71 3
81 1
%91 1
101 1
%111 1
121 1
%131 1
141 1
};
\addlegendentry{1/1 (1D)}

\addplot [orange, mark=square*,thick, mark size=3pt]
table{%
x  y
21 2
%31 4
41 2
%51 2
61 2
%71 2
81 4
%91 2
101 2
%111 2
121 2
%131 2
141 2
};
\addlegendentry{1/3 (1D)}

\addplot [red, mark=triangle*,thick, mark size=4pt]
table{%
x  y
21 4
%31 4
41 2
%51 4
61 4
%71 2
81 4
%91 4
101 2
%111 4
121 4
%131 2
141 4
};
\addlegendentry{1/4 (1D)}

\addplot [blue, mark=diamond*,thick, mark size=4pt]
table{%
x  y
21 4
%31 4
41 4
%51 4
61 4
%71 4
81 4
%91 4
101 4
%111 4
121 4
%131 4
141 4
};
\addlegendentry{1/5 (1D)}

\end{axis}

\end{tikzpicture}%
    \end{minipage}%
    }
    \hfill
    \scalebox{0.47}{
    \begin{minipage}{0.5\textwidth}
        \centering
        % This file was created with tikzplotlib v0.10.1.
\begin{tikzpicture}

\definecolor{black}{RGB}{0,0,0}
\definecolor{orange}{RGB}{250,149,0}
\definecolor{sky-blue}{RGB}{86,180,233}
\definecolor{bluish-green}{RGB}{0,158,115}
\definecolor{yellow}{RGB}{240, 228, 66}
\definecolor{blue}{RGB}{0,114,178}
\definecolor{vermillion}{RGB}{213, 94, 0}
\definecolor{reddish-purple}{RGB}{204, 121, 167}

\begin{axis}[
name=ax1,
xmajorgrids=true,
grid style={help lines,dashed},
legend cell align={left},
legend columns=2,
legend style={at={(1,1)}, draw opacity=1, text opacity=1, font=\Large},
tick align=outside,
tick pos=left,
xlabel={$N+1$},
xmin=17, xmax=151,
xtick style={color=black},
ylabel={$N_t^*$},
ymin=-1.24530827928549, ymax=60,
ytick style={color=black},
xtick={21,41,61,81,101,121,141},
xticklabels={21,41,61,81,101,121,141},
label style={font=\Large},
tick label style={font=\Large}  
]
\addplot [black, mark=*,thick, mark size=3pt]
table{%
x  y
21 1
%31 1
41 1
%51 1
61 1
%71 1
81 1
%91 1
101 1
%111 1
121 1
141 1
};
\addlegendentry{1/1 (2D)}

\addplot [orange, mark=square*,thick, mark size=3pt]
table{%
x  y
21 12
%31 12
41 12
%51 12
61 12
%71 12
81 12
%91 12
101 12
%111 12
121 12
%131 12 
141 12
};
\addlegendentry{1/3 (2D)}
\addplot [red, mark=triangle*,thick, mark size=4pt]
table{%
x  y
21 29
%31 21
41 29
%51 21
61 25
%71 29
81 13
%91 29
101 21
%111 17
121 21
%131 21
141 24
};
\addlegendentry{1/4 (2D)}
\addplot [blue, mark=diamond*,thick, mark size=4pt]
table{%
x  y
21 40
%31 40
41 42
%51 44
61 40
%71 48
81 40
%91 40
101 36
%111 52
121 44
%131 44
141 40
};
\addlegendentry{1/5 (2D)}

\end{axis}

\end{tikzpicture}%
    \end{minipage}
    }
    \caption{Optimal number of transmit antennas, i.e., number of positions where non-zero powers are allocated by the optimal distribution, for increasing $N$ for the height-to-width ratios  given in Table \ref{tab: Considered Environments}.}. 
    \label{fig: number antennas}
    \vspace*{-5mm}
\end{figure}
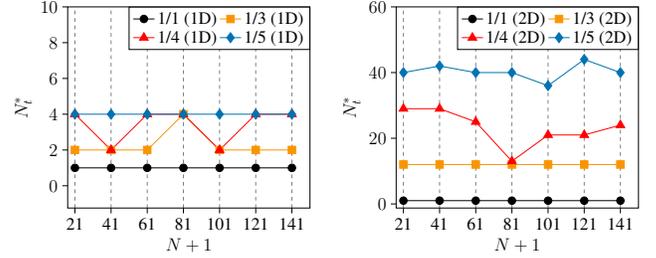
In Fig. \ref{fig: number antennas}, one can observe that the optimal number of transmit antennas $N_t^*$ does not increase monotonically for different LoDs, which suggests that a finite number of antennas is indeed optimal, even for a two-dimensional array. 
The optimal number of antennas to be deployed depends on the geometry of the environment and the dimensionality of the transmit antenna {\color{black}architecture}.
{\color{black} We note that in Section \ref{section: Methods}, the optimal solution of the problem is characterised in the context of a continuous architecture and a slight mismatch between the optimal positions regarding the continuous architecture and the optimal positions regarding the discretised architecture may occur depending on how well the optimal solution can be embedded in the grid at a certain granularity (see Section V-\ref{Subsection: Fundamental Analysis --> Heatmaps of optimal distributions}).}
%Moreover, one can assume that the embedding of the optimal deployment of antennas in the sampling grid depends on how well the geometrical arrangement can be represented within the grid at a certain granularity (see Section V-\ref{Subsection: Fundamental Analysis --> Heatmaps of optimal distributions}). 
Consequently, slight fluctuations regarding the optimal number of transmit antennas are expected, which are visible in Fig. \ref{fig: number antennas}.
However, note that these fluctuations are negligible compared to the difference in the number of grid cells among different LoDs, especially for the two-dimensional array with $(N+1)^2$ cells. 

In general, the optimal number of transmit antennas is larger when employing a two-dimensional antenna array compared to the deployment over a one-dimensional transmit antenna array. 
Moreover, for a decreasing height-to-width ratio, the optimal number of transmit antennas increases, especially when employing a two-dimensional antenna array. 

\section{PERFORMANCE AND SENSITIVITY ANALYSIS}\label{Section: Simulation Results}
{\color{black}We assume a LoD of $N+1=81$, a transmit power of $10$ W, $S_{RX}=\lambda^2/4$, which is on the scale of the antenna aperture proposed in \cite{Collado13}, and a wavelength of $\lambda = 12.5$ mm for the considered system\footnote{{\color{black}A wavelength of $\lambda = 12.5$ mm corresponds to an operating frequency of $24$ GHz, which is at the low end of the mmWave frequency band \cite{Wagih20}.}
The electromagnetic field (EMF) radiation exposure limit for the proposed system operating at {\color{black}$24$} GHz is around $9.9$ W/m$^2$ \cite{Alhasnawi20}. 
Considering the worst-scenario, i.e., when only a single transmit antenna is deployed which radiates a power of $10$ W, the minimum distance $d_\mathrm{min}$ from the single antenna to be below the exposure limit is around $d_\mathrm{min} =\sqrt{10\;\text{W} / (4 \pi \, 9.9\;\text{W/m}^2)} \approx 0.28$ m. 
The exposure to EMF radiation decreases for an increasing distance between the transmit antenna array and the receiver \cite{Alhasnawi20}.
When the transmit power is spread across multiple antennas, which is optimal for the majority of considered environments as shown in Fig. \ref{fig: MRT heatmaps}, the required distance from the ceiling to be below the exposure limit reduces. 
In our system, the transmit antenna array is positioned on the ceiling of the environment and thus, from a practical point of view, the EMF exposure to humans in the environment will typically be below the required limit since they will have a certain distance to the ceiling.}.

Next, we evaluate the boundaries of the near-field region in the context of the discretised antenna array.
The room is contained within $d_\mathrm{Fraunhofer}$ since the antenna {\color{black}architecture} covers the entire ceiling, which has a side length of at least the height of the room.
Reactive near-field effects are typically limited to a short distance around the individual radiating elements in the context of discrete antenna {\color{black}architectures} \cite{LiuHanzo23,bjornson20}.
Specifically, reactive near-field effects are negligible beyond a distance of $d_\mathrm{Fresnel} = \sqrt[3]{\frac{L_\mathrm{element}^4}{8 \lambda}}$ from an antenna element with aperture $L_\mathrm{element}$, which is typically on the order of the wavelength.
Thus, $d_\mathrm{Fresnel}$ is typically limited to a few wavelengths from the transmit antenna \cite{LiuHanzo23, Liu23}.}

\subsection{HEATMAPS OF OPTIMAL DISTRIBUTIONS}\label{Subsection: Fundamental Analysis --> Heatmaps of optimal distributions}
The solution $p^*(a_{xi},a_{zi})$ of Problem \eqref{Problem: Epigraph General Problem Discrete} is visualised in Fig. \ref{fig: MRT heatmaps} using heatmaps for the height-to-width ratios considered in Table \ref{tab: Considered Environments}. The value zero is depicted using white colour, while the non-zeros are represented by a colour corresponding to the percentage of power allocated to the position. 
Thus, transmit antennas are represented by a colour different from white.

\begin{figure}[t]
    \centering
    \scalebox{0.47}{
    \begin{minipage}{\textwidth}
        \centering
        % This file was created with tikzplotlib v0.10.1.

\begin{tikzpicture}[spy using outlines={rectangle, magnification=4,size=2cm, connect spies}]

\definecolor{darkslategray38}{RGB}{38,38,38}
\definecolor{lightgray204}{RGB}{204,204,204}

\begin{groupplot}[group style={group size=2 by 4}]
\nextgroupplot[
tick align=outside,
title= {\Large 1/1 (1 non-zero)},
x grid style={lightgray204},
xmajorticks=false,
xmin=0, xmax=81,
xtick style={color=darkslategray38},
xtick={0.5,8.5,16.5,24.5,32.5,40.5,48.5,56.5,64.5,72.5,80.5},
xticklabels={0,8,16,24,32,40,48,56,64,72,80},
y dir=reverse,
y grid style={lightgray204},
ymajorticks=false,
ymin=0, ymax=81,
ytick style={color=darkslategray38},
ytick={0.5,9.5,18.5,27.5,36.5,45.5,54.5,63.5,72.5},
yticklabel style={rotate=90.0},
yticklabels={0,9,18,27,36,45,54,63,72}
]
\addplot graphics [includegraphics cmd=\pgfimage,xmin=0, xmax=81, ymin=81, ymax=0] {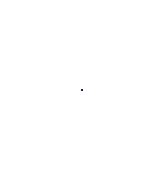};
\coordinate (centre-fig1) at (40.5,40.5);
\coordinate (centre-fig1-viewer) at (55,40.5);
\spy[size=1cm] on (centre-fig1) in node [fill=white] at (centre-fig1-viewer);

\nextgroupplot[
tick align=outside,
title={\Large 1/3 (12 non-zeros)},
x grid style={lightgray204},
xmajorticks=false,
xmin=0, xmax=81,
xtick style={color=darkslategray38},
xtick={0.5,8.5,16.5,24.5,32.5,40.5,48.5,56.5,64.5,72.5,80.5},
xticklabels={0,8,16,24,32,40,48,56,64,72,80},
y dir=reverse,
y grid style={lightgray204},
ymajorticks=false,
ymin=0, ymax=81,
ytick style={color=darkslategray38},
ytick={0.5,9.5,18.5,27.5,36.5,45.5,54.5,63.5,72.5},
yticklabel style={rotate=90.0},
yticklabels={0,9,18,27,36,45,54,63,72}
]
\addplot graphics [includegraphics cmd=\pgfimage,xmin=0, xmax=81, ymin=81, ymax=0] {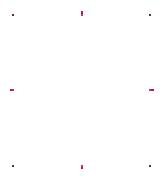};
\coordinate (spypoint3-fig2) at (6,40.5);
\coordinate (spyviewer3-fig2) at (16,40.5);
\spy[size=1cm] on (spypoint3-fig2) in node [fill=white] at (spyviewer3-fig2);

\coordinate (top-left-fig2) at (6,6);
\coordinate (viewer-top-left-fig2) at (15,15);
\spy[size=1cm] on (top-left-fig2) in node [fill=white] at (viewer-top-left-fig2);

\coordinate (top-fig2) at (40.5,6);
\coordinate (viewer-top-fig2) at (40.5,16);
\spy[size=1cm] on (top-fig2) in node [fill=white] at (viewer-top-fig2);

\nextgroupplot[
tick align=outside,
title={\Large 1/4 (13 non-zeros)},
x grid style={lightgray204},
xmajorticks=false,
xmin=0, xmax=81,
xtick style={color=darkslategray38},
xtick={0.5,8.5,16.5,24.5,32.5,40.5,48.5,56.5,64.5,72.5,80.5},
xticklabels={0,8,16,24,32,40,48,56,64,72,80},
y dir=reverse,
y grid style={lightgray204},
ymajorticks=false,
ymin=0, ymax=81,
ytick style={color=darkslategray38},
ytick={0.5,9.5,18.5,27.5,36.5,45.5,54.5,63.5,72.5},
yticklabel style={rotate=90.0},
yticklabels={0,9,18,27,36,45,54,63,72}
]
\addplot graphics [includegraphics cmd=\pgfimage,xmin=0, xmax=81, ymin=81, ymax=0] {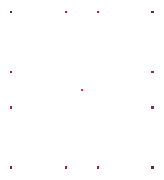};
\coordinate (centre-fig3) at (40.5,40.5);
\coordinate (centre-fig3-viewer) at (40.5,55);
\spy[size=1cm] on (centre-fig3) in node [fill=white] at (centre-fig3-viewer);

\coordinate (spypoint2-fig3) at (5,48);
\coordinate (spyviewer2-fig3) at (15,48);
\spy[size=1cm] on (spypoint2-fig3) in node [fill=white] at (spyviewer2-fig3);

\coordinate (spypoint3-fig3) at (5,33);
\coordinate (spyviewer3-fig3) at (15,33);
\spy[size=1cm] on (spypoint3-fig3) in node [fill=white] at (spyviewer3-fig3);

\coordinate (top-left-fig3) at (5,5);
\coordinate (viewer-top-left-fig3) at (15,15);
\spy[size=1cm] on (top-left-fig3) in node [fill=white] at (viewer-top-left-fig3);

\coordinate (top-fig3) at (33,5);
\coordinate (viewer-top-fig3) at (33,15);
\spy[size=1cm] on (top-fig3) in node [fill=white] at (viewer-top-fig3);

\nextgroupplot[
tick align=outside,
title={\Large 1/5 (40 non-zeros)},
x grid style={lightgray204},
xmajorticks=false,
xmin=0, xmax=81,
xtick style={color=darkslategray38},
xtick={0.5,8.5,16.5,24.5,32.5,40.5,48.5,56.5,64.5,72.5,80.5},
xticklabels={0,8,16,24,32,40,48,56,64,72,80},
y dir=reverse,
y grid style={lightgray204},
ymajorticks=false,
ymin=0, ymax=81,
ytick style={color=darkslategray38},
ytick={0.5,9.5,18.5,27.5,36.5,45.5,54.5,63.5,72.5},
yticklabel style={rotate=90.0},
yticklabels={0,9,18,27,36,45,54,63,72}
]
\addplot graphics [includegraphics cmd=\pgfimage,xmin=0, xmax=81, ymin=81, ymax=0] 
%{Figures/Heatmaps_280124/test.png};
{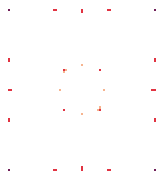};
\coordinate (spypoint1) at (49,49);
\coordinate (spyviewer1) at (65,65);
\spy[size=1cm] on (spypoint1) in node [fill=white] at (spyviewer1);

\coordinate (centre) at (40.5,51);
\coordinate (centre-viewer) at (40.5,65);
\spy[size=1cm] on (centre) in node [fill=white] at (centre-viewer);

\coordinate (left) at (32,49);
\coordinate (left-viewer) at (15,65);
\spy[size=1cm] on (left) in node [fill=white] at (left-viewer);

\coordinate (spypoint2) at (5,41);
\coordinate (spyviewer2) at (15,41);
\spy[size=1cm] on (spypoint2) in node [fill=white] at (spyviewer2);

\coordinate (spypoint3) at (5,27);
\coordinate (spyviewer3) at (15,27);
\spy[size=1cm] on (spypoint3) in node [fill=white] at (spyviewer3);

\coordinate (top-left) at (4.5,4);
\coordinate (viewer-top-left) at (15,15);
\spy[size=1cm] on (top-left) in node [fill=white] at (viewer-top-left);

\coordinate (top) at (27,4);
\coordinate (viewer-top) at (27,15);
\spy[size=1cm] on (top) in node [fill=white] at (viewer-top);

\nextgroupplot[
tick align=outside,
title={\Large 1/1 (1D) (1 non-zero)},
x grid style={lightgray204},
xmajorticks=false,
xmin=0, xmax=81,
xtick style={color=darkslategray38},
xtick={0.5,8.5,16.5,24.5,32.5,40.5,48.5,56.5,64.5,72.5,80.5},
xticklabels={0,8,16,24,32,40,48,56,64,72,80},
y dir=reverse,
y grid style={lightgray204},
ymajorticks=false,
ymin=0, ymax=81,
ytick style={color=darkslategray38},
ytick={0.5,9.5,18.5,27.5,36.5,45.5,54.5,63.5,72.5},
yticklabel style={rotate=90.0},
yticklabels={0,9,18,27,36,45,54,63,72}
]
\addplot graphics [includegraphics cmd=\pgfimage,xmin=0, xmax=81, ymin=81, ymax=0] {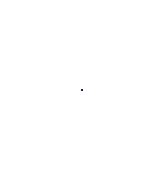};
\coordinate (centre-fig5) at (40.5,40.5);
\coordinate (centre-fig5-viewer) at (55,40.5);
\spy[size=1cm] on (centre-fig5) in node [fill=white] at (centre-fig5-viewer);

\nextgroupplot[
tick align=outside,
title={\Large 1/3 (1D) (4 non-zeros)},
x grid style={lightgray204},
xmajorticks=false,
xmin=0, xmax=81,
xtick style={color=darkslategray38},
xtick={0.5,8.5,16.5,24.5,32.5,40.5,48.5,56.5,64.5,72.5,80.5},
xticklabels={0,8,16,24,32,40,48,56,64,72,80},
y dir=reverse,
y grid style={lightgray204},
ymajorticks=false,
ymin=0, ymax=81,
ytick style={color=darkslategray38},
ytick={0.5,9.5,18.5,27.5,36.5,45.5,54.5,63.5,72.5},
yticklabel style={rotate=90.0},
yticklabels={0,9,18,27,36,45,54,63,72}
]
\addplot graphics [includegraphics cmd=\pgfimage,xmin=0, xmax=81, ymin=81, ymax=0] {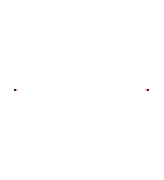};
\coordinate (centre-fig6) at (8,40.5);
\coordinate (centre-fig6-viewer) at (20,40.5);
\spy[size=1cm] on (centre-fig6) in node [fill=white] at (centre-fig6-viewer);

\nextgroupplot[
tick align=outside,
title={\Large 1/4 (1D) (4 non-zeros)},
x grid style={lightgray204},
xmajorticks=false,
xmin=0, xmax=81,
xtick style={color=darkslategray38},
xtick={0.5,8.5,16.5,24.5,32.5,40.5,48.5,56.5,64.5,72.5,80.5},
xticklabels={0,8,16,24,32,40,48,56,64,72,80},
y dir=reverse,
y grid style={lightgray204},
ymajorticks=false,
ymin=0, ymax=81,
ytick style={color=darkslategray38},
ytick={0.5,9.5,18.5,27.5,36.5,45.5,54.5,63.5,72.5},
yticklabel style={rotate=90.0},
yticklabels={0,9,18,27,36,45,54,63,72}
]
\addplot graphics [includegraphics cmd=\pgfimage,xmin=0, xmax=81, ymin=81, ymax=0] {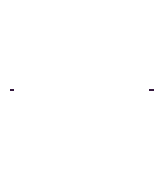};
\coordinate (centre-fig7) at (6,40.5);
\coordinate (centre-fig7-viewer) at (20,40.5);
\spy[size=1cm] on (centre-fig7) in node [fill=white] at (centre-fig7-viewer);

\nextgroupplot[
tick align=outside,
title={\Large 1/5 (1D) (4 non-zeros)},
x grid style={lightgray204},
xmajorticks=false,
xmin=0, xmax=81,
xtick style={color=darkslategray38},
xtick={0.5,8.5,16.5,24.5,32.5,40.5,48.5,56.5,64.5,72.5,80.5},
xticklabels={0,8,16,24,32,40,48,56,64,72,80},
y dir=reverse,
y grid style={lightgray204},
ymajorticks=false,
ymin=0, ymax=81,
ytick style={color=darkslategray38},
ytick={0.5,9.5,18.5,27.5,36.5,45.5,54.5,63.5,72.5},
yticklabel style={rotate=90.0},
yticklabels={0,9,18,27,36,45,54,63,72}
]
\addplot graphics [includegraphics cmd=\pgfimage,xmin=0, xmax=81, ymin=81, ymax=0] {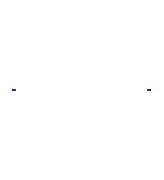};
\coordinate (centre-fig8) at (7,40.5);
\coordinate (centre-fig8-viewer) at (20,40.5);
\spy[size=1cm] on (centre-fig8) in node [fill=white] at (centre-fig8-viewer);

\end{groupplot}

\end{tikzpicture}%
    \vspace*{-40mm}
    \end{minipage}
    }
    \scalebox{0.47}{
    \begin{minipage}{\textwidth}
        \centering
        \includegraphics[scale=0.55,angle=-90]{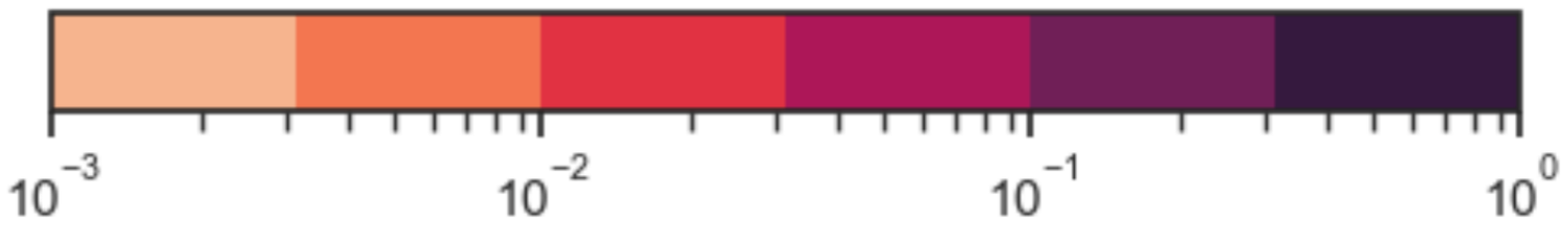}
        \vspace*{-40mm}
    \end{minipage}
    }
    \caption{Heatmaps of optimal power distributions for the height-to-width ratios listed in Table \ref{tab: Considered Environments} with $N+1=81$. Each heading indicates the height-to-width ratio from Table \ref{tab: Considered Environments}, the dimensionality of the transmit antenna array, and the number of non-zeros of $p^*(a_{ix},a_{iz})$ in brackets, i.e, the optimal number of transmit antennas. 
    {\color{black} The relative amount of power allocated to an antenna is colour-coded according to the colour bar at the bottom. The antennas are partly magnified to aid visualising the solution. The colours of the unmagnified antennas follow from the symmetry of the optimal solution discussed in Section \ref{section: Methods}.}} 
    \label{fig: MRT heatmaps}
    \vspace*{-5mm}
\end{figure}

This depiction supports our conjecture that the optimal transmit antenna positions are isolated points.
{\color{black} Except for a height-to-width ratio of one, the room is contained within the near-field of the optimal antenna deployment since the largest distance between individual transmit antennas defines the aperture of the array generated by the antennas \cite{LiuHanzo23}.}
{\color{black} From a practical point of view, the proposed optimal deployment may be realised via a distributed antenna system.}

For a height-to-width ratio equal to one, it is optimal to place a single antenna in the centre of the ceiling. 
{\color{black}
In this case, the proposed system does not have to be optimised for operating in the Fresnel region since the critical receiver positions, which are in the corners of the room at $y=L_y$, are in the far-field in this scenario. 
%We note that for larger height-to-width ratios, $d_\mathrm{Fraunhofer}$ decreases.
%We note that for a height-to-width ratio equal to one deploying a single transmit antenna is optimal and thus, the optimal solution coincides with M-FF.
}

%This is discussed further in the next section. 

\subsection{MINIMUM RECEIVED POWER AND COMPARISON TO ALTERNATIVE DEPLOYMENTS}\label{Subsection: Performance Analysis}
{\color{black} In this section, the performance of the proposed optimal two-dimensional distribution of power is measured in terms of the minimum received power in the environment.
This value corresponds to the guaranteed level of received power, which is available at any single receiver position $(x,y,z)$ beyond the reactive near-field region of the room.}
Thus, we display the performance of the worst-case scenario, which applies to a subset of receiver positions at the floor $y=L_y$ of the room since the antennas are deployed on the ceiling.

The proposed optimal method is denoted by M-OPT. Moreover, the performance is compared to other power allocation methods, which are listed below:
\begin{figure}[t]
    \centering
    \scalebox{0.47}{
    \begin{minipage}[t]{0.5\textwidth}
        \centering
        % This file was created with tikzplotlib v0.10.1.
\begin{tikzpicture}

\definecolor{black}{RGB}{0,0,0}
\definecolor{orange}{RGB}{250,149,0}
\definecolor{sky-blue}{RGB}{86,180,233}
\definecolor{bluish-green}{RGB}{0,158,115}
\definecolor{yellow}{RGB}{240, 228, 66}
\definecolor{blue}{RGB}{0,114,178}
\definecolor{vermillion}{RGB}{213, 94, 0}
\definecolor{reddish-purple}{RGB}{204, 121, 167}

\begin{axis}[
xmajorgrids=true,
ymajorgrids=true,
grid style={help lines,dashed},
legend cell align={left},
legend columns=2,
legend style={at={(1.0,1.0)},
  anchor=north east, draw opacity=1, text opacity=1, font=\large},
tick align=outside,
tick pos=left,
xmin=-0.45, xmax=8.45,
xtick style={color=black},
xtick={0,2,4,6,8,12,16,20,24,28,32},
xticklabels={$\frac{1}{1}$,$\frac{1}{2}$,$\frac{1}{3}$,$\frac{1}{4}$,$\frac{1}{5}$,$\frac{1}{7}$,$\frac{1}{9}$,$\frac{1}{11}$,$\frac{1}{13}$,$\frac{1}{15}$,$\frac{1}{17}$},
xlabel = {$\frac{\mathrm{Height}}{\mathrm{Width}}$},
minor xtick={},
ymin=-0.115862963482696, ymax=1.05,%.986687776686339,
ytick style={color=black},
ylabel = {Minimum Received Power $\gamma_{xyz}$ [µW]},
ymin=0.90042078428093e-7, ymax=0.34490105275476e-05,
ymode=log,
log basis y={10},
ytick style={color=black},
enlarge y limits=0.01,
yticklabels={0.01, 0.1, 1, 10},
label style={font=\Large},
tick label style={font=\Large}  
]

\addplot [blue, mark=square*, thick, mark size=3pt]
table{%
x  y
0 1.64910780655923e-06
1 1.1724503174397e-06
2 9.48877740506483e-07
3 7.95623132632254e-07
4 6.72080721601913e-07
5 5.81015885007686e-07
6 5.07765435157197e-07
7 4.48715774595826e-07
8 4.00367886474475e-07
9 3.59853919575093e-07
10 3.25702396599195e-07
11 2.96605136651589e-07
12 2.71509535806819e-07
13 2.49643566090357e-07
14 2.30484936403532e-07
15 2.13628231790075e-07
16 1.98706533259948e-07
17 1.85421074948639e-07
18 1.73510247069322e-07
19 1.62778834741349e-07
20 1.53067047683671e-07
21 1.44245907497596e-07
22 1.36202154546718e-07
23 1.2885711809396e-07
24 1.22125370539916e-07
25 1.15945953703343e-07
26 1.10241082856691e-07
27 1.04997761105945e-07
28 1.00130986863299e-07
29 9.56325473679241e-08
30 9.14434802887377e-08
31 8.75503970352052e-08
32 8.39116264322527e-08
33 8.05065242579001e-08
34 7.73202183090065e-08
35 7.43312388975852e-08
36 7.15237199285336e-08
37 6.88780468536366e-08
};
\addlegendentry{M-OPT}
\addplot [black, mark=*, thick, mark size=3pt]
table{%
x  y
0 1.64910780667867e-06
1 1.16407609883201e-06
2 8.24553903339337e-07
3 5.99675566064973e-07
4 4.49756674548729e-07
5 3.47180590879721e-07
6 2.74851301113112e-07
7 2.22351614383642e-07
8 1.83234200742075e-07
9 1.53405377365458e-07
10 1.30192721579895e-07
11 1.11803919096859e-07
12 9.70063415693338e-08
13 8.4932590901906e-08
14 7.49594457581216e-08
15 6.66306184516636e-08
16 5.96063062654943e-08
17 5.36295221684122e-08
18 4.85031707846669e-08
19 4.40741507352875e-08
20 4.02221416263091e-08
21 3.68515710989648e-08
22 3.38857768495618e-08
23 3.12627072356147e-08
24 2.89317159066434e-08
25 2.68511447491779e-08
26 2.49864819193739e-08
27 2.33089442640095e-08
28 2.17943762997182e-08
29 2.04223876988071e-08
30 1.91756721706823e-08
31 1.80394655242881e-08
32 1.70011114090585e-08
33 1.60497110139044e-08
34 1.51758387117669e-08
35 1.43713098621235e-08
36 1.36289901378403e-08
37 1.29426381165102e-08
};
\addlegendentry{M-FF}
\addplot [orange, mark=diamond*, thick, mark size=4pt]
table{%
x  y
0 1.58012481692979e-06
1 1.15622318903056e-06
2 8.70747557184224e-07
3 6.77469108257657e-07
4 5.42387491859245e-07
5 4.44710828669664e-07
6 3.71878585494047e-07
7 3.16111736237135e-07
8 2.72432153717257e-07
9 2.37548971728327e-07
10 2.0922077184568e-07
11 1.85879512388425e-07
12 1.66401826388194e-07
13 1.49965742847168e-07
14 1.35958416331537e-07
15 1.23915191071802e-07
16 1.13478428280768e-07
17 1.04369104318901e-07
18 9.63668463565722e-08
19 8.92956566994651e-08
20 8.3013544238976e-08
21 7.7404885378158e-08
22 7.2374721699085e-08
23 6.784445173313e-08
24 6.3748539603618e-08
25 6.00319745212839e-08
26 5.66482910342579e-08
27 5.35580125041026e-08
28 5.07274171637244e-08
29 4.81275523126037e-08
30 4.57334410302923e-08
31 4.35234394626236e-08
32 4.14787127674466e-08
33 3.95828052378158e-08
34 3.7821285674367e-08
35 3.61814532645097e-08
36 3.46520924062145e-08
37 3.32232673484934e-08
};
\addlegendentry{M-UNI}
\addplot [red, mark=triangle*, mark size=4pt]
table{%
x  y
0 1.64910780664503e-06
1 1.17245031745648e-06
2 9.48877740404362e-07
3 7.95623128472639e-07
4 6.72080721566586e-07
5 5.81015879318377e-07
6 5.07205435040222e-07
7 4.468715774590978e-07
8 3.96367882104983e-07
9 3.522853913553283e-07
10 3.25702399846939e-07
%11 2.96605115653708e-07
%12 2.7150952944077e-07
% 13 2.49643160444507e-07
% 14 2.30484473054055e-07
% 15 2.13627951260579e-07
% 16 1.98705817961116e-07
% 17 1.85403705127679e-07
% 18 1.7347239437413e-07
% 19 1.62774126584269e-07
% 20 1.51745704886993e-07
% 21 1.42422291919131e-07
% 22 1.24792995369178e-07
% 23 1.1655424846041e-07
% 24 9.43130029650674e-08
% 25 8.51176239976912e-08
% 26 7.49480417213076e-08
% 27 6.22513401301184e-08
% 28 4.53598067685923e-08
% 29 4.64753954254723e-08
% 30 3.70982775798313e-08
% 31 3.59032723077349e-08
% 32 3.35012907400849e-08
% 33 3.05799298914385e-08
% 34 2.8810954777746e-08
% 35 2.77115774211564e-08
% 36 2.65962378947409e-08
% 37 2.51442878229315e-08
};
\addlegendentry{M-S75}
% \addplot [blue, mark=pentagon*, dash dot]
% table{%
% x  y
% 0 1.64910780667355e-06
% 1 1.17245031735329e-06
% 2 9.48877740404445e-07
% 3 7.95623128474213e-07
% 4 6.72080721565912e-07
% 5 5.81015879249482e-07
% 6 5.07765435033529e-07
% 7 4.48715774389443e-07
% 8 4.00367836838716e-07
% 9 3.59853898841309e-07
% 10 3.14000824547438e-07
% 11 2.72995455158027e-07
% 12 2.58249222418528e-07
% 13 1.60360342295127e-07
% 14 1.36618033597113e-07
% 15 1.24846588266618e-07
% 16 9.83540674917242e-08
% 17 8.43596513099723e-08
% 18 7.90764396996306e-08
% 19 7.12271206893123e-08
% 20 6.45560146898107e-08
% 21 5.75673362597902e-08
% 22 5.45506800999449e-08
% 23 4.8099087112095e-08
% 24 4.37093056119856e-08
% 25 4.03139654096167e-08
% 26 3.826253297117e-08
% 27 3.48228864697436e-08
% 28 2.9985028231248e-08
% 29 2.90505462320479e-08
% 30 2.63878496950552e-08
% 31 2.51333364436762e-08
% 32 2.28191488996598e-08
% 33 2.18135667693593e-08
% 34 2.02911093372965e-08
% 35 1.95000330864056e-08
% 36 1.85168125007627e-08
% 37 1.75975158070432e-08
% };
% \addlegendentry{M-S90}
\end{axis}

\end{tikzpicture}%
    \end{minipage}%
    }
    \scalebox{0.47}{
    \begin{minipage}[t]{0.5\textwidth}
        \centering
        % This file was created with tikzplotlib v0.10.1.
\begin{tikzpicture}

\definecolor{black}{RGB}{0,0,0}
\definecolor{orange}{RGB}{250,149,0}
\definecolor{sky-blue}{RGB}{86,180,233}
\definecolor{bluish-green}{RGB}{0,158,115}
\definecolor{yellow}{RGB}{240, 228, 66}
\definecolor{blue}{RGB}{0,114,178}
\definecolor{vermillion}{RGB}{213, 94, 0}
\definecolor{reddish-purple}{RGB}{204, 121, 167}

\begin{axis}[
xmajorgrids=true,
grid style={help lines,dashed},
legend cell align={left},
legend columns=2,
legend style={at={(0.0,0.0)},
  anchor=south west, draw opacity=1, text opacity=1, font=\large},
tick align=outside,
tick pos=left,
xmin=-0.45, xmax=8.45,
xtick style={color=black},
xtick={0,2,4,6,8,12,16,20,24,28,32},
xticklabels={$\frac{1}{1}$,$\frac{1}{2}$,$\frac{1}{3}$,$\frac{1}{4}$,$\frac{1}{5}$,$\frac{1}{7}$,$\frac{1}{9}$,$\frac{1}{11}$,$\frac{1}{13}$,$\frac{1}{15}$,$\frac{1}{17}$},
xlabel = {$\frac{\mathrm{Height}}{\mathrm{Width}}$},
ymin=0.315862963482696, ymax=1.05,%.986687776686339,
ytick style={color=black},
ylabel = {Loss},
label style={font=\Large},
tick label style={font=\Large}  
]
\addplot [black, mark=*, thick, mark size=3pt]
table{%
0 1.0000000000724305
1 0.9928575066396182
2 0.8689780233428334
3 0.7537181128469643
4 0.6692003803899164
5 0.5975406178010577
6 0.5412958072422202
7 0.4955288558418109
8 0.45766458033281027
9 0.42629903141417935
10 0.3997290868574991
11 0.3769453231964462
12 0.35728521019001924
13 0.3402154208579333
14 0.3252249232760397
15 0.3118998734078326
16 0.299971547425254
17 0.28923099590090956
18 0.27954067038639285
19 0.27076094263311307
20 0.2627746613982679
21 0.2554774117219159
22 0.24879031438477606
23 0.24261529124699538
24 0.23690176560968795
25 0.23158328420738739
26 0.22665308859361633
27 0.22199467891977617
28 0.217658658747389
29 0.21355059821042618
30 0.20969971954407293
31 0.20604664439196366
32 0.20260733979199644
33 0.1993591347017994
34 0.19627257971669773
35 0.19334145475396364
36 0.19055203157020287
37 0.1879065784779416
};
\addlegendentry{M-FF}
\addplot [orange, mark=diamond*, thick, mark size=4pt]
table{%
x  y
0 0.998
1 0.9861596451740767
2 0.9176604319113275
3 0.8514949860951204
4 0.8070273025633851
5 0.7654021863167841
6 0.7323826313205389
7 0.7044809969559629
8 0.6804545592210611
9 0.6601261200901161
10 0.6423679224661757
11 0.6266901325002037
12 0.6128765455464167
13 0.6007194385009251
14 0.5898798353290295
15 0.5800506329779922
16 0.5710855421764879
17 0.5628761689997259
18 0.5553957070793121
19 0.5485704381736919
20 0.5423345226500493
21 0.5366175493016876
22 0.5313772160209123
23 0.5265091501088744
24 0.5219925992591546
25 0.5177582537712393
26 0.513858260154232
27 0.510087186050202
28 0.5066105783315449
29 0.5032549444431719
30 0.5001279575743011
31 0.497124409899846
32 0.49431425096896486
33 0.4916720179225915
34 0.4891513048141699
35 0.4867597231140075
36 0.48448392282782266
37 0.48234914992714073
};
\addlegendentry{M-UNI}
\addplot [red, mark=triangle*, thick, mark size=4pt]
table{%
x  y
0 1.0000000000520262
1 1.000000000014316
2 0.9999999998923774
3 0.9999999947718777
4 0.99999999999474368
5 0.9999999902079981
6 0.9989999997696265
7 0.9959999999891966
8 0.9899999890863058
9 0.9789999832659607
10 1.0000000099715067
11 0.999999929205942
12 0.9999999765531312
13 0.9999983750999217
14 0.9999979896756646
15 0.9999986868332243
16 0.9999964002248931
17 0.999906322293918
18 0.9997818417307854
19 0.9999710763559183
20 0.9913675554819059
21 0.9873575922526907
22 0.9162336365712442
23 0.9045231663137198
24 0.772263801928379
25 0.7341146566914393
26 0.6798558194383562
27 0.5928825479174331
28 0.4530046910505164
29 0.48597885034546817
30 0.40569625590245995
31 0.4100869159199555
32 0.3992449218837706
33 0.37984412037807824
34 0.37261864241774967
35 0.37281199442051444
36 0.3718519942938619
37 0.36505518044613305
};
\addlegendentry{M-S75}
% \addplot [blue, mark=pentagon*, dash dot]
% table{%
% x  y
% 0 1.0000000000693205
% 1 0.9999999999263011
% 2 0.9999999998924652
% 3 0.9999999947738555
% 4 0.9999999999464346
% 5 0.9999999900894209
% 6 0.9999999997564457
% 7 0.9999999995400584
% 8 0.9999998760246236
% 9 0.9999999423827749
% 10 0.9640728094913075
% 11 0.9204002946135917
% 12 0.9511607820738729
% 13 0.642357200734288
% 14 0.5927417024682359
% 15 0.5844105304831622
% 16 0.49497148321267004
% 17 0.4549625835862481
% 18 0.45574507001904907
% 19 0.43756991381889426
% 20 0.42174991722073657
% 21 0.3990916432811065
% 22 0.400512607759327
% 23 0.37327458369061156
% 24 0.35790520363415845
% 25 0.3476961819018124
% 26 0.34708052551433716
% 27 0.3316536095908433
% 28 0.29945803162994955
% 29 0.30377258612889046
% 30 0.2885700501745356
% 31 0.2870727865867903
% 32 0.27194263619813347
% 33 0.2709540247878572
% 34 0.26242954017801695
% 35 0.2623396754260087
% 36 0.25889051239595245
% 37 0.255488019926542
% };
% \addlegendentry{M-S90}
\end{axis}

\end{tikzpicture}%
    \end{minipage}
    }
    \caption{Minimum received power at the worst location in the environment and loss of the benchmark schemes and suboptimal schemes compared to proposed optimal deployment and power allocation. The height is fixed at $2$ m.}. 
    \label{fig: Performance Analysis}
    \vspace*{-5mm}
\end{figure}
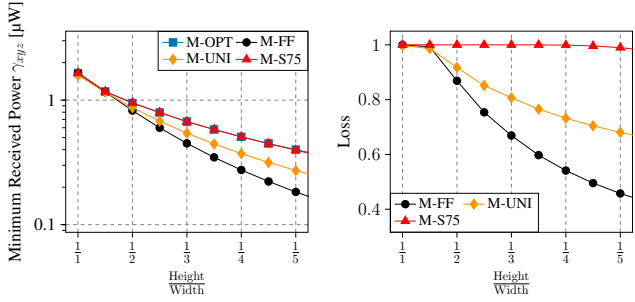
\begin{itemize}
\item {\color{black}The achievable performance gain of modelling the system in the Fresnel region is highlighted by comparing the proposed approach to the optimal power distribution designed under the far-field approximation.
By virtue of the max-min problem, under the far-field assumption, the ideal deployment location lies in the centre of the ceiling since the path losses corresponding to different transmit antennas are assumed to be identical. As a consequence of the constant path loss, a single transmit antenna is deployed to which all the available power is allocated.}
This benchmark scheme is referred to as M-FF. 
\item A uniform power distribution across the antenna array, where an antenna is located at every grid cell and the power is split equally among the antennas. This benchmark scheme is referred to as M-UNI.
\item {\color{black}We derive a suboptimal scheme from the proposed optimal approach by removing optimally deployed transmit antennas operating with low transmit power.
Thereby, small non-zero values are set to zero in the distribution, and subsequently, the allocated powers are re-normalised for a fair comparison at equal total transmit power.}
More specifically, the x-th percentile of the allocated power levels was calculated and transmit antennas operating at a power level lower than the one corresponding to the x-th percentile were removed, followed by the re-normalisation of power. {\color{black}We investigated the 75th, percentile of the non-zero values. The scheme is referred to as M-S75.} 

%{\color{black}The suboptimal scheme M-S75 is derived from the optimal distribution by applying a non-linear transformation.} 
%Similar to Section IV-\ref{Subsection: Fundamental Analysis --> Sparsity of optimal distribution}, 
%the efficiency of embedding the antenna deployment within the grid depends on the geometry of the environment and the LoD.
%Thus, for a fixed LoD, certain geometries allow for a more efficient representation of the optimal transmit antenna deployment than others.
%Naturally, these fluctuations apply to the suboptimal scheme as well since it is based on the optimal transmit antenna deployment.
%Specifically, if more grid cells are required for representing the optimal transmit antenna deployment, the respective transmit powers per antenna decrease since the overall transmit power remains constant.
%Consequently, this impacts the number and the position of removed antennas in the suboptimal scheme.

\end{itemize}

{\color{black}We present the performance} results in Fig. \ref{fig: Performance Analysis}.
{\color{black}Hereby, we investigate the performance by evaluating the minimum received power\footnote{
{\color{black}Typical EH circuits are designed for receive powers in the µW range \cite{Clerckx19}.}
The constantly available and guaranteed received power is in the range of, e.g., the continuous standby power required by the nanowatt sensors presented in \cite{Olsson19}.
Moreover, receivers capable of receiving energy at multiple positions, e.g., through multiple receiving antennas can increase the amount of received power.
}, i.e., the power at the critical receiver positions, for M-OPT, M-FF, M-UNI, and M-S75 for environments with different height-to-width ratios.
Moreover, we determine the performance loss incurred by utilising M-FF, M-UNI, and M-S75 in comparison to the proposed optimal scheme M-OPT.}
{\color{black}To this end, we calculate} the loss compared to the proposed optimal solution by dividing the objective value of a benchmark scheme{\color{black}, i.e., M-FF or M-UNI,} or the suboptimal scheme{\color{black}, i.e., M-S75,} by the objective value obtained through the optimal deployment and power allocation, i.e., the solution of Problem \eqref{Problem: Epigraph General Problem Discrete}{\color{black}, i.e., M-OPT}. 

{\color{black}For all schemes, we observe that the minimum received power decreases for a decreasing height-to-width ratio.
This effect occurs due to the increasing environment size for decreasing height-to-width ratio while keeping the total transmit power constant.}
Note that the received power can be increased by adding more receive antennas {\color{black}or increasing the receiver's antenna size $S_{RX}$}.

{\color{black}We observe that} the proposed optimal power distribution outperforms or matches the performance of the other schemes.
{\color{black} For a height-to-width ratio of one, we see from Fig. \ref{fig: MRT heatmaps} that the proposed optimal solution coincides with M-FF and we find no performance loss between M-OPT and the other schemes.
On the other hand, for decreasing height-to-width ratio the performance gap between M-OPT and the alternative methods grows.
We observe the largest performance gap between M-OPT and M-FF which indicates the relevance of accurately modelling the wireless channel by taking into account the spherical EM wavefronts when operating in the near-field region.
Moreover, a uniform distribution of power, i.e., in the case of M-UNI, is also found to be highly suboptimal, especially for small height-to-width ratios.}
%{\color{black}Moreover, the performance loss compared to the optimal solution increases For all schemes, we observe that the minimum received power decreases for a decreasing height-to-width ratio.
%This effect occurs since the environment increases for decreasing height-to-width ratio while the total transmit power is kept constant.}
%``{\color{black} In Fig. 5, we include additional environments with lower height-to-width ratios to achieve a detailed investigation of the performance of the proposed solution and the alternative deployments.}"
%However, the loss compared to the proposed optimal solution increases for decreasing height-to-width ratio.
%The largest loss occurs when utilising the optimal distribution for the far-field operating regime. 
%Especially, for small height-to-width ratios, a uniform distribution of power across all antenna elements is also highly suboptimal. 
{\color{black}On the other hand, the loss in performance compared to the optimal distribution is negligible for M-S75 with a slight loss observable for a height-to-width ratio of $1/5$.
This is because antennas radiating a relatively low amount of power contribute little to the minimum received power. Therefore, removing these transmit antennas comes at a small loss in received power for the investigated environments.}
%Depending on the geometry of the environment, M-S75 and M-S90 show a significant loss in receiver power. Especially for low height-to-width ratios, a larger number of transmit antennas operating with lower powers is important for covering the indoor space adequately.

\subsection{SENSITIVITY ANALYSIS}\label{Subsection: Sensitivity Analysis}
{\color{black}
While the assumption of a strong LoS channel, which was made throughout this paper, is valid for very high carrier frequencies, a low-power small-scale fading component may still be present. This requires including an NLoS component in the wireless channel model \cite{Tse2005a}. 
Assuming isotropic scattering, the NLoS component can be modelled as spatially correlated Rayleigh fading \cite{Bjornson2021}, \cite{Pizzo2020}, \cite{Bjornson2017}.
The superposition of the NLoS and LoS components is modelled by a Rician fading channel.
%Hereby, the correlated fading model is necessary for modelling the channel accurately in the near-field region \cite{Liu23}.
}
%{\color{black} So far, we have optimised the transmit power distribution across a continuous aperture antenna assuming an LoS channel. 
%Hereby, the robustness of the solution in the presence of a low-power small-scale fading component is assessed using a Rician channel model in Section V-\ref{Subsection: Sensitivity Analysis}.}

In this section, we assess the robustness of the proposed optimal two-dimensional power distribution in the presence of scattering objects in the environment using a Rician fading channel model.
The Rician fading channel between transmit antenna $i$ and the receiver position {\color{black}$(x,L_y,z)$} is defined as 
\begin{align}\label{eq: Rician Fading Channel}
    \Tilde{g}_{xz}(a_{xi},a_{zi}) =
    \sqrt{\frac{\kappa}{\kappa+1}}\, g_{xz}(a_{xi},a_{zi})
    + \sqrt{\frac{1}{\kappa+1}}  \,h(a_{xi},a_{zi}),
\end{align}
where constant $\kappa$ denotes the Rician $K$-factor and the LoS channel is given in \eqref{eq: LOS Channel}. The NLoS channel vector $\boldsymbol{h} = [h(a_{x1},a_{z1}),\dots,h(a_{x(N+1)^2},a_{z(N+1)^2})]^T \in \mathbb{C}^{{(N+1)^2} \times 1}$ is a realisation of the circularly invariant complex Gaussian distribution $\boldsymbol{h} \sim \mathcal{N}_\mathbb{C}\left(\boldsymbol{0},\,\frac{{c}}{ \Tilde{D}^2 } A \boldsymbol{R} \right)$, where $\Tilde{D}$ is the average distance between transmitter and receiver \cite{Bjornson2021}. $A$ is the area of a transmit antenna and the entries of the spatial correlation matrix $\boldsymbol{R} \in \mathbb{C}^{{(N+1)^2} \times {(N+1)^2}}$ are given by
\begin{align}\label{eq: Spatial Correlation Matrix}
    R_{ij} = \mathrm{sinc}\left( \frac{2}{\lambda} \sqrt{(a_{xi} - a_{xj})^2 + (a_{zi} - a_{zj})^2} \right), \nonumber \\ \forall i,j=1,\dots,(N+1)^2,
\end{align}
with $\mathrm{sinc}(x) = \mathrm{sin}(\pi x)/(\pi x)$ \cite{Bjornson2021}. 
%In case of small-scale fading, the propagation environment is no longer deterministic. 
%Consequently, the objective is to maximise the expected value of the received power, i.e., $\mathrm{E}_H\left\{ \gamma_{xyz}\right\}$ where $H$ denotes the random variable describing the Rician fading channel\footnote{$\mathrm{E}_H\left\{ \gamma_{xyz}\right\}$ was calculated by averaging over the solutions of the $100$ channel realisations.}.
%For the following simulations, $\kappa = 10$. 
%All simulation results involving the Rician fading channel were averaged over 100 channel realisations.
{\color{black}The propagation environment is no longer deterministic in case of small-scale fading. 
For the simulations}, the coefficient of variation (CoV) is provided to indicate the variability of the simulation values in relation to the mean. This dimensionless measure is comparable across the simulation results of different environments.
{\color{black}For the simulations, we consider $\kappa=10$, and $A=(\lambda/2)^2$.}
\vspace*{-3mm}
\subsubsection{Discreteness of the Optimal Solution}\label{Subsection: Sensitivity Analysis --> Sparsity of optimal solution}
We investigate the sparsity of the transmit antenna deployment at LoD $N+1=81$ by calculating the number of antennas relative to the number of available grid cells $(N+1)^2$.
{\color{black}The optimal distribution of transmit power for an individual channel realisation is determined by solving \eqref{Problem: Epigraph General Problem Discrete}, whereby we utilise the squared magnitude of the Rician channel \eqref{eq: Rician Fading Channel}, after normalisation by the reference channel gain $c$, instead of the normalised squared magnitude of the LoS channel, i.e., $f_{xz}(a_{xi},a_{zi})$, in the constraints \eqref{Constraint: Epigraph General Problem Discrete}.
This yields the optimal number of transmit antennas for an individual channel realisation.
Subsequently, we average the optimal number of transmit antennas over the different channel realisations.}
%were averaged and the optimal number of transmit antennas corresponds to the number of transmit antennas of the averaged distribution.}
The results are given in Table \ref{tab: Fading sparsity} {\color{black} and we list the percentage of non-zeros of the proposed deployment developed under LoS conditions for comparison}.
% OLD VALUES
% \begin{table}[h]
% \renewcommand{\arraystretch}{1.3}
% \caption{Percentage of non-zeros in the optimal distribution of power under perfect LoS and in the fading environment at LoD $N+1=81$.}
% \label{tab: Fading sparsity}
% \centering
% \begin{tabular}{|c|c|c|c|}
% \hline
% $\frac{\mathrm{Height}}{\mathrm{Width}}$ & LoS [\%] & Rician Mean [\%] & Rician CoV [\%] \\ \hhline{|=|=|=|=|}
% 1/1 & 0.152 & 0.207 & 69.029 \\ \hline
% 1/3 & 0.183 & 0.688 & 13.674 \\ \hline
% 1/4 & 0.198 & 1.089 & 11.819 \\ \hline
% 1/5 & 0.61 & 1.714 & 10.102 \\ \hline
% \end{tabular}
% \end{table}
% VALUES AFTER MAJOR REVISION
\begin{table}[h]
\renewcommand{\arraystretch}{1.3}
\caption{{\color{black} Average number of transmit antennas relative to number of grid cells under fading conditions compared with the percentage of non-zeros in the proposed distribution developed under LoS conditions} at LoD $N+1=81$.}
\label{tab: Fading sparsity}
\centering
\begin{tabular}{|c|c|c|c|}
\hline
$\frac{\mathrm{Height}}{\mathrm{Width}}$ & LoS [\%] & Rician Mean [\%] & Rician CoV [\%] \\ \hhline{|=|=|=|=|}
1/1 & 0.152 & 0.155 & 77.10 \\ \hline
1/3 & 0.183 & 0.65 & 16.81 \\ \hline
1/4 & 0.198 & 1.14 & 10.87 \\ \hline
1/5 & 0.61 & 1.751 & 10.45 \\ \hline
\end{tabular}
\end{table}
{\color{black}Table \ref{tab: Fading sparsity} reveals that} the distribution of power{\color{black}, when optimised under fading conditions,} is less sparse across all environments compared to the distribution of power under perfect LoS conditions, i.e., a larger number of transmit antennas {\color{black}is required}.
Moreover, the CoV is similar across different environments, except for the environment with a height-to-width ratio of one. This system is especially sensitive to channel fluctuations since under LoS conditions, the optimal deployment requires only a single transmit antenna.

\subsubsection{Performance Comparison}\label{Subsection: Sensitivity Analysis --> Performance comparison}
{\color{black} Next, the performance of the proposed optimal deployment and power allocation is investigated in the small-scale fading environment.
For each channel realisation, we determine the minimum received power which is achieved when utilising the proposed deployment and power allocation, which was optimised for LoS conditions.
As a benchmark, we optimise the deployment and power allocation for each channel realisation and then determine the received power at the critical receiver positions.
We average the minimum received power of the optimal LoS solution and the benchmark scheme across the channel realisations, respectively.
%Subsequently, we average the received power of the  across the channel realisations to obtain a performance upper-bound.
%The performance of the proposed optimal solution in LoS conditions is compared to the one for a small-scale fading environment. 
The loss in performance due to small-scale fading is measured by the ratio of the averaged objective value of the proposed scheme developed under LoS conditions divided by the averaged objective value of the benchmark scheme}.
Simulations are performed for environments with different height-to-width ratios and the results are displayed in Table \ref{tab: Fading performance}.  
\begin{table}[h]
\renewcommand{\arraystretch}{1.3}
\caption{{\color{black}Performance loss due to small-scale fading.}}
\label{tab: Fading performance}
\centering
\begin{tabular}{|c|c|c|}
\hline
$\frac{\mathrm{Height}}{\mathrm{Width}}$ & \makecell{Average Relative \\Performance [\%]} & CoV [\%] \\ \hhline{|=|=|=|}
1/1 & 98.752 & 0.529 \\ \hline
1/3 & 98.77 & 0.153 \\ \hline
1/4 & 98.79 & 0.126 \\ \hline
1/5 & 98.8 & 0.131 \\ \hline
\end{tabular}
\end{table}

The performance loss caused by small-scale fading was found to be below {\color{black}2}\% {\color{black}for all considered} environments. 

\section{CONCLUSION}\label{Section: Conclusion}
In this paper, we formulated an optimisation problem for {\color{black}determining the optimal distribution of transmit power across a continuous aperture antenna}.
The proposed solution maximises the received powers at the worst possible receiver positions in an indoor space.
This system is especially relevant for environments with a multitude of mobile low-power devices, which require a reliable supply of power, such as wearables.
In this scenario, focusing energy beams to every mobile device may not be possible in a reliable manner.
Mathematically, the formulated problem bears a resemblance to that of finding the optimal amplitude-constrained input distribution for achieving the information capacity. 
Consequently, analogous theoretical results are obtained, which guarantee that the optimal transmit antenna positions are spread out within their area of deployment such that they do not form any regions of concentration and do not occupy any measurable area. 
{\color{black} Specifically, a transmit antenna architecture with a spatially continuous aperture, which is utilised in HMIMO systems, was proven to be unnecessary for the considered WPT problem.}
Furthermore, the optimal number of transmit antennas is shown to be finite analytically for a one-dimensional transmit antenna array.
{\color{black}Having excluded a continuous aperture antenna as a suitable transmit antenna architecture, the problem was discretised and the resulting solution yielded the optimal transmit antenna deployment including the optimal allocation of transmit power.}
In the environments investigated through simulation, deploying a finite number of antennas {\color{black}was} found to be optimal in all cases. 
{\color{black}Based on our results, the optimal deployment may be achieved by a finite, and surprisingly small, number of transmit antennas radiating a fixed amount of power.}
Moreover, our proposed method yields a considerable gain compared to benchmark schemes.
{\color{black}Finally}, the robustness of the approach has been validated by determining the optimal solution in a fading environment with a strong LoS component. The loss in performance under fading conditions was found to be around {\color{black}2\% in all considered environments}.

Possible extensions of the proposed approach offer several opportunities for future work. These include the extension of the method to more geometrically challenging environments and more flexible {\color{black}transmit antenna architecture} placements, such as positioning {\color{black}transmitters} on multiple walls as well as the ceiling of a room. 
{\color{black}Furthermore, a small-scale multipath fading component for modelling the wireless channel may be directly incorporated into the analysis of the problem to improve the robustness of the solution.}
{\color{black}Moreover, an investigation of the impact of the reactive near-field region when the deployment of large continuous antenna arrays is affordable presents another interesting topic for future work.  }
{\color{black}
While we characterise the solution to our problem through tools from topology theory and functional analysis, determining the optimal distribution of power and the transmit antenna deployment through different mathematical techniques may provide additional insights.
Furthermore, a possible extension to our problem is the joint optimisation of the amplitudes and the phases of the transmit signals of the individual radiating elements. } 

\appendices
    
\section{}\label{app: omega set}
By definition, any $P \in \Omega$ is non-negative and bounded by one (see \eqref{eq: Finite Power}).
Observe that given any two distributions $P_1,P_2 \in \Omega$, their convex combination $\lambda P_1 + (1- \lambda) P_2 \in \Omega$ with coefficient $\lambda \in [0,1]$ is a distribution as well. It follows that $\Omega$ is a convex set \cite{boyd2004convex}.

In order to further characterise the set $\Omega$, the weak topology on a Hilbert space is defined.
A linear functional in $\Omega$ is a linear map from the Hilbert space to the underlying field.
The weak topology on a Hilbert space is defined to be the weakest topology such that all continuous linear functionals on the Hilbert space are continuous \cite{Rudin91}.

Observe that $\mathcal{L}^2(\mathcal{X}_a \times \mathcal{Z}_a)$ is an infinite-dimensional vector space equipped with the following inner product \cite{Vetterli2014}
\begin{align}\label{eq: Inner Product L2 Hilbert Space}
    \iint_{\mathcal{X}_a \times \mathcal{Z}_a} w(a_x,a_z)\,v(a_x,a_z) \,\mathrm{d}a_x\,\mathrm{d}a_z, \nonumber \\ \forall w(a_x,a_z), v(a_x,a_z) \in \mathcal{L}^2(\mathcal{X}_a \times \mathcal{Z}_a),
\end{align}
where $w(a_x,a_z)$ and $v(a_x,a_z)$ are arbitrary elements of $\mathcal{L}^2(\mathcal{X}_a \times \mathcal{Z}_a)$.
This inner product space is complete since completeness is ensured by Lebesgue measurability and Lebesgue integration \cite{Vetterli2014}. Thus, $\mathcal{L}^2(\mathcal{X}_a \times \mathcal{Z}_a)$ is a Hilbert space, which is separable since it has a countable orthonormal basis \cite{Vetterli2014}.
Note that the set $\Omega$ is a subset of $\mathcal{L}^2(\mathcal{X}_a \times \mathcal{Z}_a)$ with bounded elements in $\Omega$ (see \eqref{eq: Finite Power}).
In a separable Hilbert space, every bounded sequence has a weakly convergent subsequence in the weak topology \cite{Rudin91} and therefore, $\Omega$ is sequentially compact \cite{Willard04}. 
The concepts of sequential compactness and compactness are equivalent for a separable Hilbert space \cite{Rudin91}. 
Thus, $\Omega$ is a compact space.

\section{}\label{app: Lagrangian}
    First, Problem \eqref{Problem: General Problem} is rewritten in its epigraph form \cite{boyd2004convex} which yields the following constrained optimisation problem
    \begin{subequations}\label{Problem: Epigraph General Problem}
    \begin{alignat}{2}
    &\underset{o \in \mathbb{R}, P \in \Omega}{\text{maximise}}
    &\qquad& o \label{Objective: Epigraph General Problem}\\
    &\text{subject to} 
    && o \leq \Phi_{P}\{ f_{xz} \}, \quad \forall (x,L_y,z) \in \mathcal{X}\times\mathcal{Y_{\mathrm{crit}}}\times\mathcal{Z},
    \label{Constraint: Epigraph General Problem}
    \end{alignat}
    \end{subequations}
    where $o$ is bounded by $m = \underset{\substack{(x,L_y,z) \\ \in \mathcal{X}\times\mathcal{Y_{\mathrm{crit}}}\times\mathcal{Z}}}{\min} \,\Phi_{P^*}\{ f_{xz} \}$. 
    Observe that the objective function and the feasible set of Problem \eqref{Problem: Epigraph General Problem} are both affine.
    Then, if the linear programming problem in \eqref{Problem: Epigraph General Problem} admits a solution, i.e., if $o^*=m$ is finite, duality holds.
    Furthermore, there exist some $\lambda_{xz} \geq 0, \forall (x,L_y,z) \in \mathcal{X}\times\mathcal{Y_{\mathrm{crit}}}\times\mathcal{Z}$ such that the augmented problem
    \begin{align}\label{Problem: Lagrangian General Problem, proof}
        \underset{o \in \mathbb{R}, P \in \Omega, \Lambda}{\text{maximise}}\quad \!
        o - \sum_{\mathcal{X}\times\mathcal{Y_{\mathrm{crit}}}\times\mathcal{Z}} \left( o - \Phi_{P}\{ f_{xz} \} \right) \lambda_{xz},
    \end{align}
    % {\color{black}\begin{align}\label{Problem: Lagrangian General Problem, proof}
    %     \underset{o \in \mathbb{R}, P \in \Omega, \Lambda}{\text{maximise}}\quad \!
    %     o - \iint_{\mathcal{X}\times\mathcal{Y_{\mathrm{crit}}}\times\mathcal{Z}} \left( o - \mathrm{E}_{P}\{ f_{xz} \} \right) \lambda_{xz} \,\mathrm{d}x\,\mathrm{d}y\,\mathrm{d}z,
    % \end{align}}
    and Problem \eqref{Problem: Epigraph General Problem} are equivalent if the solution $(o^*,P^*)$ is feasible, i.e., the problem constraints are satisfied in $(o^*,P^*)$, and the following complementary slackness conditions hold:
    \begin{align}\label{eq: comp slack}
       \lambda_{xz} \left( m - \Phi_{P^*}\{ f_{xz} \} \right) &\overset{!}{=} 0, \quad \forall (x,L_y,z) \in \mathcal{X}\times\mathcal{Y_{\mathrm{crit}}}\times\mathcal{Z}.
    \end{align}
    Since $f_{xz}$ is bounded, the optimal objective value $o^*=m$ is finite and it is attained in $P^*$.
    The conditions in \eqref{eq: comp slack} imply
    \begin{align}\label{eq: Complementary Slackness, proof}
        [&(\lambda_{xz} > 0 \implies m = \Phi_{P^*}\{ f_{xz} \})  \,\,\, \mathrm{or} \,\,\, \nonumber \\ 
       &(m < \Phi_{P^*}\{ f_{xz} \} \implies \lambda_{xz} = 0)], \, \nonumber \\ &\forall (x,L_y,z) \in \mathcal{X}\times\mathcal{Y_{\mathrm{crit}}}\times\mathcal{Z}
    \end{align}
    From \eqref{eq: Complementary Slackness, proof}, if $\lambda_{xz} > 0$, then $o^* =m= \Phi_{P^*}\{ f_{xz} \}$ and position $(x,L_y,z)$ is critical. If $o^* = m < \Phi_{P^*}\{ f_{xz} \}$, then $\lambda_{xz} = 0$ and position $(x,L_y,z)$ is non-critical. Thus, the dual variables $\lambda_{xz}$ are only non-zero for critical receiver locations and it is sufficient to sum over $\mathcal{R}_\mathrm{crit}$ in \eqref{Problem: Lagrangian General Problem, proof}.
    Thus, under the assumption that $\mathcal{R}_\mathrm{crit}$ is known, the modified augmented problem in \eqref{Problem: Lagrangian General Problem} along with \eqref{eq: Complementary Slackness, proof} is equivalent to the augmented problem in \eqref{Problem: Lagrangian General Problem, proof} together with \eqref{eq: Complementary Slackness, proof} and thus, to Problem \eqref{Problem: General Problem}.

\section{}\label{app: Necessary and Sufficient Constraints}
The continuity of $\Phi_P\{ f_{xz} \}$ in $P$ is established in the following lemma.
\begin{lemmaappendix}\label{lemma: continuous J}
    $\Phi_P\{ f_{xz} \}$ is weakly continuous in the weak topology on a Hilbert space if $f_{xz}$ is a continuous and bounded function.
\end{lemmaappendix}
\begin{proof}
    {\color{black}
    Based on Theorem 1 in \cite{Dytso2018}, the linear functional $\Phi_P\{ f_{xz} \}$ is weakly continuous on $\Omega$ if $f_{xz}$ is a continuous and bounded function.
    From Key Observation 1, $f_{xz}$ is a continuous function and bounded since $y=L_y$.}
\end{proof}

    {\color{black}Recall the mathematical equivalence between the WPT problem considered in this paper, where we optimise the transmit power distribution across a continuous aperture antenna and the problem of determining the amplitude-constrained capacity-achieving input distribution of a channel studied in, e.g., \cite{smith1971information,Morsi2020,Dytso2018}.}
    For the proof of Proposition \ref{proposition: Necessary and Sufficient Constraints} we apply Theorem 9 in \cite{Dytso2018}.
    The convexity and compactness of $\Omega$ is established in Lemma \ref{lemma: Omega set}.
    The continuity of $J(P)$ is shown in Lemma \ref{lemma: continuous J}.
    The concavity of $J(P)$ follows from the linearity of $J(P)$ in $P$.
    Formulating the optimality condition for $P^*$ requires defining the Gâteaux derivative of $J(P)$.
    \begin{definitionappendix}\label{definition : Gâteaux derivative}
    The Gâteaux derivative of the function $J: \Omega \rightarrow \mathbb{R}$ at $P$ in the direction of $Q$ with $P,Q \in \Omega$ is defined as follows \cite{Dytso2018}
    \begin{align}\label{eq: Gateaux derivative 1}
        \Delta_{Q} J(P) = \underset{\theta \rightarrow 0}{\mathrm{lim}} \, \frac{ J( (1-\theta)P + \theta Q) - J(P) }{\theta}.
    \end{align}
\end{definitionappendix}
    Moreover, the Gâteaux derivative of $J(P)$ in \eqref{eq: Gateaux derivative 1} can be expressed as \cite{Dytso2018}
    \begin{align}\label{eq: Gateaux derivative 2}
        \Delta_{Q} J(P) = J(Q) - J(P),
    \end{align}
    since it is a linear functional of $P$.
    Theorem 9 in \cite{Dytso2018} establishes that if $\Omega$ is a convex set and $J(P)$ is continuous, concave, and Gâteaux differentiable, then $J(P)$ attains the maximum value at $P^* \in \Omega$ if and only if $\Delta_{P}J(P^*) \leq 0, \forall P \in \Omega$.
    Applying \eqref{eq: Gateaux derivative 2} and cancelling the terms that are independent of the distributions $P$ and $P^*$, we obtain
    \begin{align}\label{eq: J_P}
        \Delta_{P}J(P^*) = \sum_{\mathcal{R}^\mathrm{crit}} \,
        \lambda_{xz} \Phi_{P}\{ f_{xz} \}
        - \sum_{\mathcal{R}^\mathrm{crit}} \,
        \lambda_{xz} \Phi_{P^*}\{ f_{xz} \}.
    \end{align}
    % {\color{black}\begin{align}\label{eq: J_P}
    %     \Delta_{P}J(P^*) = &\iint_{\mathcal{R}^\mathrm{crit}} \,
    %     \lambda_{xz} \mathrm{E}_{P}\{ f_{xz} \} \,\mathrm{d}x\, \mathrm{d}z \nonumber \\
    %     &- \iint_{\mathcal{R}^\mathrm{crit}} \,
    %     \lambda_{xz} \mathrm{E}_{P^*}\{ f_{xz} \} \,\mathrm{d}x\, \mathrm{d}z.
    % \end{align}
    % }
    We can further enforce in \eqref{eq: J_P} the constraints on $\lambda_{xz}$ and $\Phi_{P}\{ f_{xz} \}$ for the optimal solution $P^*$ in \eqref{eq: Complementary Slackness}. Then, $\lambda_{xz}$ is only non-zero for critical receiver locations w.r.t. the optimal distribution $P^*$.
    Furthermore, the objective value for the critical receiver positions is the constant $\Phi_{P^*}\{ f_{xz} \} = m$.
    By expanding the expectation operator in \eqref{eq: J_P}, the optimality condition $\Delta_{P}J(P^*) \leq 0$ can be expressed as 
    \begin{align}\label{eq: conditions}
        &\iint_{\mathcal{X}_a \times \mathcal{Z}_a} \,
        \frac{
        \sum_{\mathcal{R}^\mathrm{crit}} \,
        f_{xz}(a_x,a_z) \lambda_{xz}
        }
        {
        \sum_{\mathcal{R}^\mathrm{crit}} \,
        \lambda_{xz}
        }
        \, \,\mathrm{d}P(a_x,a_z) \nonumber \\ 
        = 
        &\iint_{\mathcal{X}_a \times \mathcal{Z}_a} \,
        \Bar{f}(a_x,a_z) \, \,\mathrm{d}P(a_x,a_z) \leq
        m,
    \end{align}
    % {\color{black}\begin{align}\label{eq: conditions}
    %     &\iint_{\mathcal{X}_a \times \mathcal{Z}_a} \,
    %     \frac{
    %     \iint_{\mathcal{R}^\mathrm{crit}} \,
    %     f_{xz}(a_x,a_z) \lambda_{xz} \,\mathrm{d}x\, \mathrm{d}z
    %     }
    %     {
    %     \iint_{\mathcal{R}^\mathrm{crit}} \,
    %     \lambda_{xz} \,\mathrm{d}x\, \mathrm{d}z
    %     }
    %     \, \,\mathrm{d}P(a_x,a_z) \nonumber \\ 
    %     = 
    %     &\iint_{\mathcal{X}_a \times \mathcal{Z}_a} \,
    %     \Bar{f}(a_x,a_z) \, \,\mathrm{d}P(a_x,a_z) \leq
    %     m,
    % \end{align}}
    which is satisfied with equality for the optimal distribution of power $P^*$.
    % {\color{black} We note that $\iint_{\mathcal{R}^\mathrm{crit}} \,
    %     \lambda_{xz} \,\mathrm{d}x\, \mathrm{d}z$ is always positive since each critical position is associated with an area with a finite size.}
    
\section{}\label{app: Optimality of distribution iff two equations hold}
    If the conditions \eqref{eq: optimality condition 1} and \eqref{eq: optimality condition 2} are satisfied, then conditions \eqref{eq: nec und suf} and \eqref{eq: bar f} in Proposition \ref{proposition: Necessary and Sufficient Constraints} are satisfied.
    In the following, we show that the converse holds as well, i.e., if Proposition \ref{proposition: Necessary and Sufficient Constraints} holds, then Corollary \ref{corollary: Optimality of distribution iff two equations hold} holds as well.
    We proceed by contradiction and assume that Proposition \ref{proposition: Necessary and Sufficient Constraints} holds, however, \eqref{eq: optimality condition 1} does not hold. 
    In this case, $\exists (a_{x1}, a_{z1}) \in \mathcal{X}_a \times \mathcal{Z}_a$, such that $\Bar{f}(a_{x1}, a_{z1}) > m$.
    Suppose, $(a_{x1}, a_{z1})$ is the only transmit antenna, i.e., $(a_{x1}, a_{z1})$ is allocated all the power.
    Then \eqref{eq: nec und suf} yields
    \begin{align}\label{eq: violation 1}
        \iint_{\mathcal{X}_a \times \mathcal{Z}_a} \, \Bar{f}(a_x,a_z) \, \,\mathrm{d}P(a_x,a_z) = \Bar{f}(a_{x1}, a_{z1}) > m,
    \end{align}
    which contradicts \eqref{eq: nec und suf} in Proposition 1. Consequently, if \eqref{eq: nec und suf} holds, then \eqref{eq: optimality condition 1} holds as well. \\
    Next, suppose Proposition \ref{proposition: Necessary and Sufficient Constraints} holds, however, \eqref{eq: optimality condition 2} does not.
    Define a subset $\mathcal{E}_0(P^*) \subset \mathcal{E}(P^*)$ such that $P^*(\mathcal{E}_0(P^*))=0$.
    Observe that $\Bar{f}(a_{x}, a_{z})$ is continuous on $\mathbb{R}^2$ since the function's numerator is a weighted sum of continuous functions and the denominator is always a positive constant.
    Assume there exists a point of increase $(a_{x1}, a_{z1}) \in \mathcal{E}(P^*)\setminus\mathcal{E}_0(P^*)$ such that $\Bar{f}(a_{x1}, a_{z1}) > m$.
    For any $\delta >0$ such that $\Bar{f}(a_{x1}, a_{z1}) - \delta >m$, there exists an open ball $\mathcal{B}$ centred at $(a_{x1}, a_{z1})$ such that $\Bar{f}(a_{x}, a_{z}) > \Bar{f}(a_{x1}, a_{z1})-\delta > m, \forall (a_{x}, a_{z}) \in \mathcal{B}$, which follows from the continuity of $\Bar{f}(a_{x}, a_{z})$.
    Since the open ball $\mathcal{B}$ contains the point of increase $(a_{x1}, a_{z1}) \in \mathcal{E}(P^*)\setminus\mathcal{E}_0(P^*)$, it follows that $P^*(\mathcal{B})>0$. 
    This contradicts Proposition \ref{proposition: Necessary and Sufficient Constraints}, i.e., $m = \iint_{\mathcal{X}_a \times \mathcal{Z}_a} \, \Bar{f}(a_x,a_z)
    \, \,\mathrm{d}P^*(a_x,a_z)$, since
    \begin{align}\label{eq: violation 2}
        m = &\underbrace{\iint_{(\mathcal{E}(P^*)\setminus\mathcal{E}_0(P^*)) \cap \mathcal{B}} \, \Bar{f}(a_x,a_z) \, \,\mathrm{d}P^*(a_x,a_z)}_{>\, m P^*(\mathcal{B})} \nonumber \\ 
        &+ 
        \underbrace{\iint_{(\mathcal{E}(P^*)\setminus\mathcal{E}_0(P^*))  \setminus \mathcal{B} } \, \Bar{f}(a_x,a_z) \, \,\mathrm{d}P^*(a_x,a_z)}_{= \,m(1-P^*(\mathcal{B}))} \nonumber \\
        &+
        \underbrace{\iint_{\mathcal{E}_0(P^*)} \, \Bar{f}(a_x,a_z) \, \,\mathrm{d}P^*(a_x,a_z)}_{=0} 
        > m.
    \end{align}
    From this contradiction we have shown that $\Bar{f}(a_x,a_z)=m, \forall (a_x,a_z) \in \mathcal{E}(P^*)\setminus\mathcal{E}_0(P^*)$. 
    Furthermore, since $P^*(\mathcal{E}_0(P^*)) = 0$, $\mathcal{E}_0(P^*)$ has no interior points.
    Any open set of $\mathbb{R}^2$ containing a point in $\mathcal{E}_0(P^*)$ intersects points in $\mathcal{E}(P^*)\setminus\mathcal{E}_0(P^*)$ and therefore, every point in $\mathcal{E}_0(P^*)$ is in the closure of $\mathcal{E}(P^*)\setminus\mathcal{E}_0(P^*)$.
    Moreover, since $\Bar{f}(a_x,a_z)$ is continuous and bounded, it must also be constant on the closure of $\mathcal{E}(P^*)\setminus\mathcal{E}_0(P^*)$. This implies $\Bar{f}(a_x,a_z) = m, \forall (a_{x}, a_{z}) \in \mathcal{E}(P^*)$.
    Consequently, if \eqref{eq: nec und suf} holds, then \eqref{eq: optimality condition 2} holds as well.
    Thus, \eqref{eq: optimality condition 1} and \eqref{eq: optimality condition 2} are necessary and sufficient conditions for $P^*(a_x,a_z)$.

\section{}\label{app: Nowhere dense set}
The proof is similar to Theorem 12 in \cite{Dytso2018}.\\
    %The characterisation of the set of optimal transmit antenna positions requires the introduction of the following definitions.
%A subset $S_1 \subset S_0$ is \emph{dense} in $S_0$ if every element $s \in S_0$ is either an element of $S_1$ or an accumulation point of $S_1$ \cite{Dytso2018}.
%A subset $S_1 \subset S_0$ is \emph{nowhere dense} if $S_2 \cap S_1$ is not \emph{dense} in $S_0$, where $S_2$ is any non-empty open set $S_2 \subset S_0$ \cite{Dytso2018}.
The characterisation of the support $\mathcal{E}$ requires the Identity Theorem for real-analytic functions, which is stated, e.g., in \cite{Krantz2002}. The Identity Theorem states that if two real-analytic functions $f_1$ and $f_2$, which are defined on an open set $U \subseteq \mathbb{R}^n$, take the same value on an open subset $W \subseteq U$, i.e., $f_1(w) = f_2(w), \forall w \in W$, then $f_1(u) = f_2(u), \forall u \in U$.

Next, we show that $\Bar{f}(a_x,a_z)$ is not constant in $\mathcal{X}_a \times \mathcal{Z}_a$.
Thereby, it suffices to show that the partial derivative of $\Bar{f}(a_x,a_z)$ w.r.t. $a_x$, with $a_z = 0$, is not zero $\forall (a_x,a_z) \in \mathcal{X}_a \times \mathcal{Z}_a$. Note that
    \begin{align}
        \frac{\partial \Bar{f}(a_x,a_z)}{\partial a_x} =
        \frac{\sum_{\mathcal{R}_\mathrm{crit}} \,
        \dfrac{2\left(x-a_x\right)}{\left(\left(z-a_z\right)^2+\left(x-a_x\right)^2+L_y^2\right)^2} \lambda_{xz} }
        {\sum_{\mathcal{R}_\mathrm{crit}} \,
        \lambda_{xz} },
    \end{align}
    % {\color{black}\begin{align}
    %     \frac{\partial \Bar{f}(a_x,a_z)}{\partial a_x} =
    %     \frac{\iint_{\mathcal{R}_\mathrm{crit}} \,
    %     \dfrac{2\left(x-a_x\right)}{\left(\left(z-a_z\right)^2+\left(x-a_x\right)^2+L_y^2\right)^2} \lambda_{xz} \,\mathrm{d}x\,\mathrm{d}z }
    %     {\iint_{\mathcal{R}_\mathrm{crit}} \,
    %     \lambda_{xz}  \,\mathrm{d}x\,\mathrm{d}z},
    % \end{align}}
    is only zero for $x=a_x$ and non-zero otherwise. Consequently, $\Bar{f}(a_x,a_z)$ is not constant in $\mathcal{X}_a \times \mathcal{Z}_a$.

    Next, we show that $\Bar{f}(a_x,a_z)$ is real-analytic in $\mathbb{R}^2$ and thus, in $\mathcal{X}_a \times \mathcal{Z}_a$.
    Let $\Bar{f}(a_x,a_z) = kn(a_x,a_z) $, where $k = \left(\sum_{\mathcal{R}_{\mathrm{crit}}} \, \lambda_{xz} \right)^{-1}$ is a positive constant for all $(a_x,a_z)$ and $n(a_x,a_z) = \sum_{\mathcal{R}_{\mathrm{crit}}} \, f_{xz}(a_x,a_z) \lambda_{xz}$. The function $\Bar{f}(a_x,a_z)$ is real-analytic in $(a_x,a_z)$ if $n(a_x,a_z)$ is an analytic function.
    The function $n(a_x,a_z)$ is real-analytic in $(a_x,a_z)$ if $f_{xz}(a_x,a_z) \lambda_{xz}$ is real-analytic for fixed $(x,z)$. Thereby, it suffices to show that $f_{xz}(a_x,a_z)$ is real-analytic since $\lambda_{xz}$ is some non-negative constant scaling $f_{xz}(a_x,a_z)$.
    
    Note that the composition of real-analytic functions is real-analytic \cite{Krantz2002}. By rewriting $f_{xz}(a_x,a_z)$ as $f(w_{xz}(a_x,a_z)) = 1 / (L_y^2 + w_{xz}(a_x,a_z))$, where $w_{xz}(a_x,a_z) = (x-a_{x})^2 + (z-a_{z})^2 \geq 0$, it is sufficient to show that both $w_{xz}(a_x,a_z)$ and $f(w) = 1/({w + L_y^2}) > 0$ are real-analytic.
    By definition, polynomials are real-analytic functions \cite{Krantz2002}, thus $w_{xz}(a_x,a_z)$ is real-analytic.
    Next, the real function $f(w)$ is shown to be analytic by proving that it's complex extension $f(r_1 + \mathrm{j}r_2) = u(r_1,r_2) + \mathrm{j} v(r_1,r_2) = (r_1+ L_y^2)/((r_1+L_y)^2 + r_2^2) - \mathrm{j} r_2/((r_1+L_y)^2 + r_2^2)$ is holomorphic. 
    Function $f(w)$ is holomorphic and thus, analytic, if it satisfies the Cauchy-Riemann (CR) equations \cite{Rudin87}.
    The CR equations are given by \cite{Rudin87}
    \begin{align}\label{eq: cr equations 1}
        \frac{\partial u(r_1,r_2)}{\partial r_1} =  \frac{\partial v(r_1,r_2)}{\partial r_2},
    \end{align}
    \begin{align}\label{eq: cr equations 2}
        \frac{\partial u(r_1,r_2)}{\partial r_2} =  - \frac{\partial v(r_1,r_2)}{\partial r_1},
    \end{align}
    Clearly, the CR equations are satisfied for the complex extension  $f(r_1 + \mathrm{j}r_2)$ since
    \begin{align}\label{eq: cr equations - solved 1}
        \frac{\partial u(r_1,r_2)}{\partial r_1} =  \frac{\partial v(r_1,r_2)}{\partial r_2} = \frac{r_2^2 - (L_y^2 +r_1)^2}{((L_y^2+r_1)^2+r_2^2)^2},
    \end{align}
    \begin{align}\label{eq: cr equations - solved 2}
        \frac{\partial u(r_1,r_2)}{\partial r_2} =  - \frac{\partial v(r_1,r_2)}{\partial r_1} = \frac{-2r_2(L_y^2+r_1)}{((L_y^2+r_1)^2+r_2^2)^2},
    \end{align}
    which proves that the complex extension $f(r_1 + \mathrm{j}r_2)$ is analytic.
    Consequently, $\Bar{f}(a_x,a_z)$ is real-analytic in $\forall(a_x,a_z) \in \mathbb{R}^2$ and thus, in $\forall(a_x,a_z) \in \mathcal{X}_a \times \mathcal{Z}_a$.
    
    Suppose that $\mathcal{E}(P^*)$ is not a nowhere dense set in $\mathcal{X}_a \times \mathcal{Z}_a$. 
    Consequently, there exists some non-empty open set $\mathcal{E}_2 \subset \mathcal{X}_a \times \mathcal{Z}_a$ such that $\mathcal{E}_2 \cap \mathcal{E}(P^*)$ is dense in $\mathcal{X}_a \times \mathcal{Z}_a$.
    From \eqref{eq: optimality condition 2}, $\Bar{f}(a_x,a_z)=m$ is constant on $\mathcal{E}(P^*)$ and thus, constant on $\mathcal{E}_2 \cap \mathcal{E}(P^*)$.
    Since $\Bar{f}(a_x,a_z)$ is real-analytic, the Identity Theorem of real-analytic functions applies and thus, $\Bar{f}(a_x,a_z)$ must be constant on $\mathcal{X}_a \times \mathcal{Z}_a$. However, $\Bar{f}(a_x,a_z) \neq \text{const.}, \forall (a_x,a_z) \in \mathcal{X}_a \times \mathcal{Z}_a$. From this contradiction, $\mathcal{E}(P^*) \subset \mathcal{X}_a \times \mathcal{Z}_a$ must be a nowhere dense set of $\mathcal{X}_a \times \mathcal{Z}_a$.
    
\section{}\label{app: Measure zero}
Suppose $\mathcal{E}(P^*)$ is of positive Lebesgue measure.
For any open subset $\mathcal{V} \subset \mathcal{E}(P^*)$, the function $\Bar{f}(a_x,a_z)$ is constant on $\mathcal{V}$ since it is constant on $\mathcal{E}(P^*)$.
Since $\Bar{f}(a_x,a_z)$ is real-analytic as shown in Appendix \ref{app: Nowhere dense set}, the Identity Theorem, e.g., \cite{Krantz2002} of real-analytic functions applies and thus, $\Bar{f}(a_x,a_z)$ must be constant on $\mathcal{X}_a \times \mathcal{Z}_a$.
However, as shown in Appendix \ref{app: Nowhere dense set}, $\Bar{f}(a_x,a_z) \neq \text{const.}, \forall (a_x,a_z) \in \mathcal{X}_a \times \mathcal{Z}_a$. From this contradiction, $\mathcal{E}(P^*) \subset \mathcal{X}_a \times \mathcal{Z}_a$ must have Lebesgue measure zero.
    
\section{}\label{app: 1D result}
The proof is also analogous to \cite{smith1971information,Morsi2020}. \\
    This proof requires the application of the Bolzano-Weierstrass theorem which states that in the set of real numbers any infinite bounded sequence of real numbers has a convergent subsequence and the limit of that subsequence is an accumulation point of the original sequence \cite{Rudin76}.
    In the following $\Bar{f}(a_l)$ is a function defined over a one-dimensional set, i.e., $\Bar{f}: \mathcal{L}_a \rightarrow \mathbb{R}$, with $\mathcal{L}_a \subset \mathcal{X}_a \times \mathcal{Z}_a$ and $\Tilde{\mathcal{E}}(\Tilde{P}^*)$ is the optimal set of transmit antenna positions, which is bounded, i.e., $\Tilde{\mathcal{E}}(\Tilde{P}^*) \subset \mathcal{L}_a$.
    Suppose the set $\Tilde{\mathcal{E}}(\Tilde{P}^*)$ is not finite.
    From the Bolzano-Weierstrass theorem, it follows that $\Tilde{\mathcal{E}}(\Tilde{P}^*)$ has an accumulation point $q$.
    Based on \eqref{eq: optimality condition 2}, define the function $s(a_l) = m - \Bar{f}(a_l)$, which is equal to zero $\forall a_l \in \Tilde{\mathcal{E}}$. Moreover, $s(a_l)$ is analytic since it is the sum of analytic functions \cite{Krantz2002}. 
    By assuming $\Tilde{\mathcal{E}}(\Tilde{P}^*)$ is not finite, $s$ has infinitely many zeros.
    Moreover, the accumulation point $q$ is also a zero of function $s(a_l)$, i.e., $s(q)=0$, since $s(a_l)$ is a continuous function, which is implied by the analycity of $s(a_l)$.
    Then, from the Identity Theorem, e.g. \cite{Krantz2002}, $s(a_l)$ coincides with zero in the entire set $\mathcal{L}_a$ or equivalently
        \begin{align}\label{eq: Incorrect}
            \Bar{f}(a_l) = m, \forall a_l \in \mathcal{L}_a,
        \end{align}
    which implies that $\Bar{f}$ is constant on $\mathcal{L}_a$.
    However, $\Bar{f}$ is not constant on $\mathcal{L}_a$ as shown in Appendix \ref{app: Nowhere dense set}.
    Consequently, $\Tilde{\mathcal{E}}(\Tilde{P}^*)$ must be a finite set.

\bibliographystyle{ieeetr}
\bibliography{refs}

\newpage

\begin{IEEEbiography}[{\includegraphics[width=1in,height=1.25in,clip,keepaspectratio]{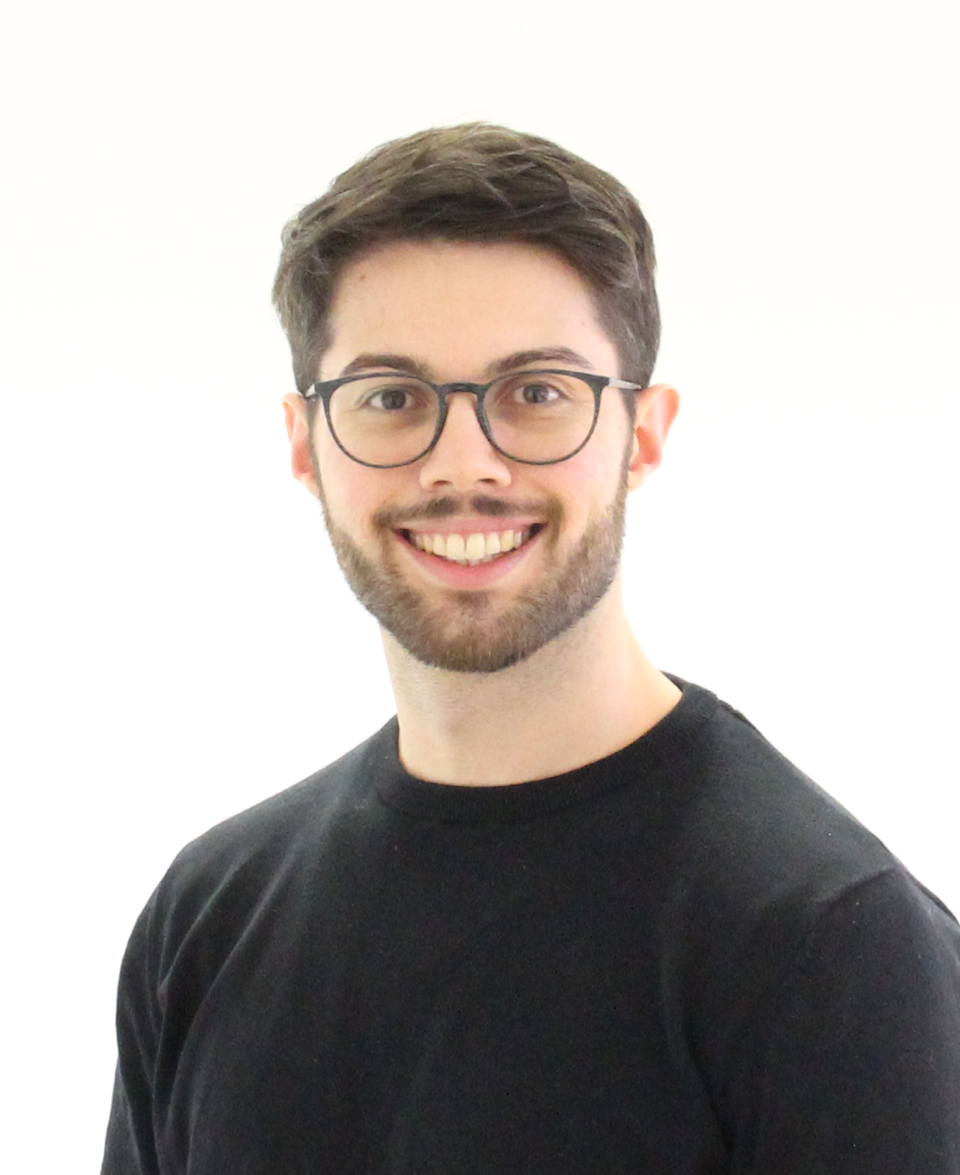}}]{KENNETH M. MAYER } $\,$ (Graduate Student Member, IEEE) received the B.Sc. degree in Medical Engineering (2019) and the M.Sc. degree in Advanced Signal Processing and Communications Engineering (2021) from Friedrich-Alexander University (FAU) Erlangen-Nuremberg, Germany. He is currently pursuing the Ph.D. degree at the Institute for Digital Communications at FAU focusing on wireless communications and wireless power transfer.
\end{IEEEbiography}
%\vspace*{-mm}
\begin{IEEEbiography}[{\includegraphics[width=1in,height=1.25in,clip,keepaspectratio]{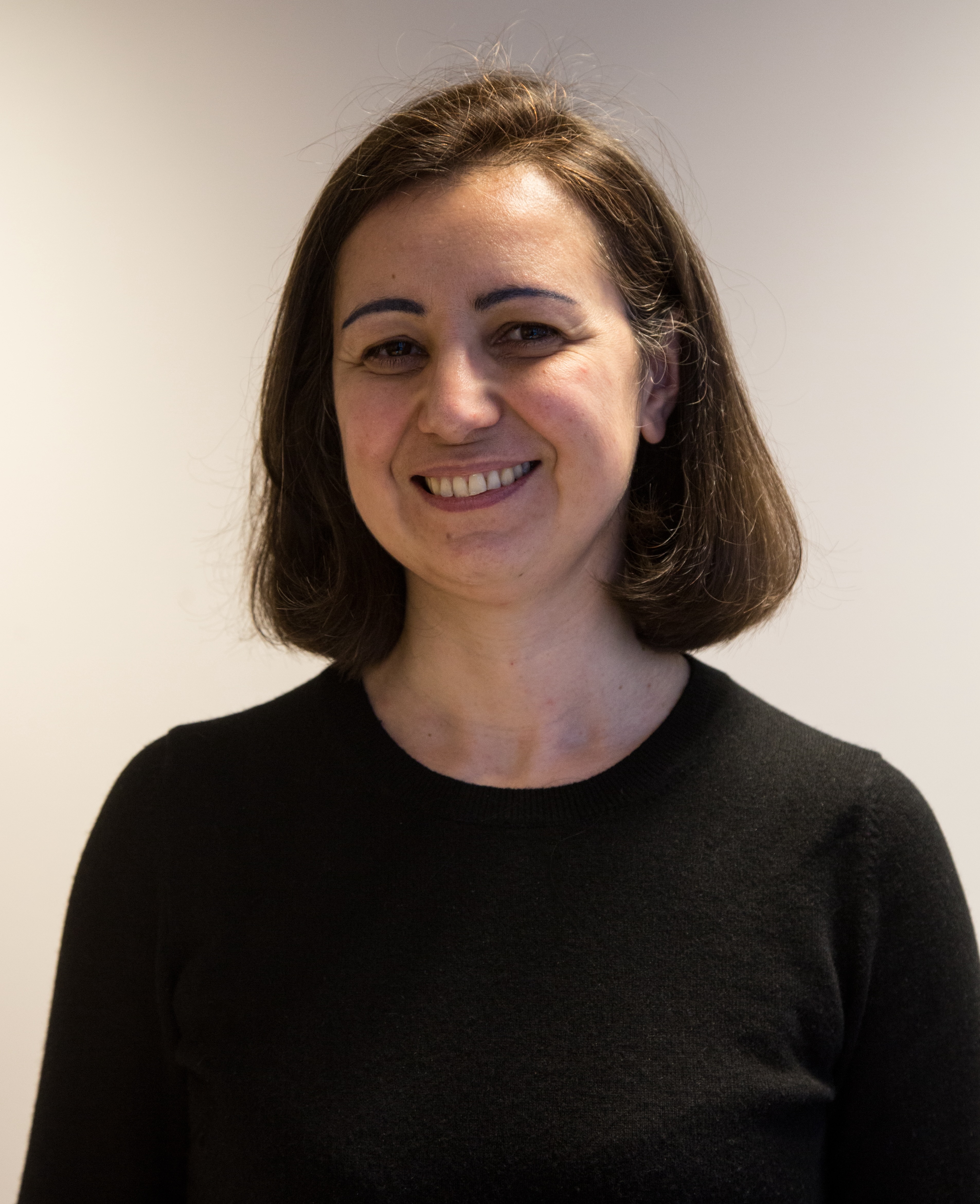}}]{LAURA COTTATELLUCCI} $\,$ (Member, IEEE) received the M.Sc. degree in electrical engineering from La Sapienza University of Rome, Italy, the Ph.D. degree from Technical University of Vienna, Austria (2006), and the Habilitation degree from University of Nice-Sophia Antipolis, France. Since December 2017 and September 2021, she has been a professor for digital communications at Friedrich Alexander Universität (FAU) Erlangen-Nuremberg, Germany, and an Adjunct Professor at EURECOM, France, respectively. She worked for 5 years (1995-2000) at Telecom Italia in the design of telecommunications networks and as Senior Researcher at Forschungszentrum Telekommunikation Wien, Austria (2000-2005). She was a Research Fellow at INRIA, France, in the quarter October-December 2005, and at the University of South Australia in 2006. From December 2006 to November 2017, she was an Assistant Professor at EURECOM, France, where she was also an Adjunct Professor (2018-2019). She was an Elected Member of the IEEE Technical Committee on Signal Processing for Communications and Networking (2017-22), served as an Associate Editor for the IEEE TRANSACTIONS ON COMMUNICATIONS (2015-2020) and the IEEE TRANSACTIONS ON SIGNAL PROCESSING (2016- 2020). Since September 2022 she has been senior member of the editorial board of IEEE \textit{Signal Processing Magazine}.  Her research interests are in the field of communication and information theory and signal processing for wireless communications, satellite, and complex networks.
\end{IEEEbiography}
%\vspace*{-1mm}
\newpage
\begin{IEEEbiography}[{\includegraphics[width=1in,height=1.25in,clip,keepaspectratio]{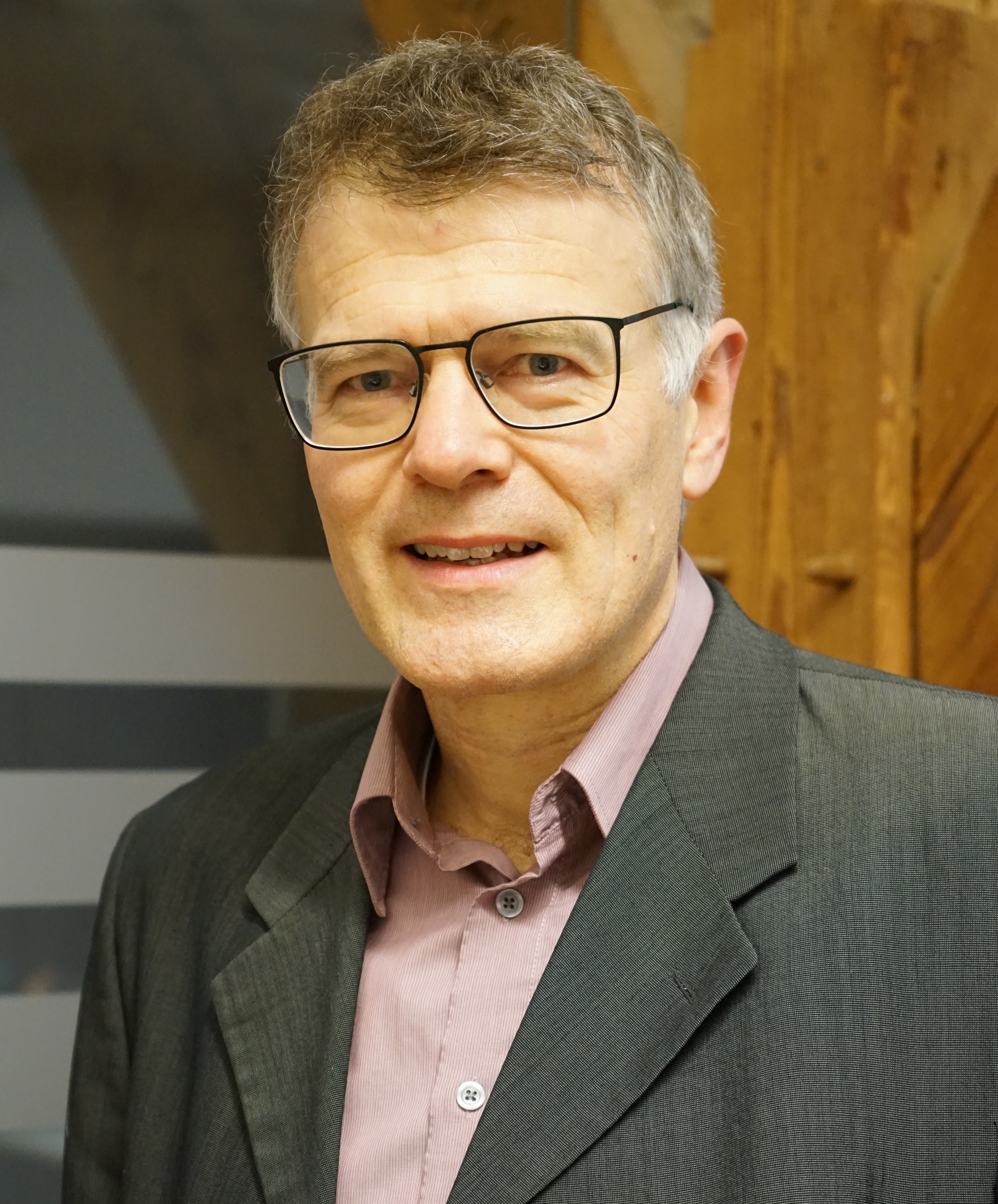}}]{ROBERT SCHOBER} $\,$ (Fellow, IEEE) received the Diplom (Univ.) and the Ph.D. degrees in electrical engineering from Friedrich-Alexander University of Erlangen-Nuremberg (FAU), Germany, in 1997 and 2000, respectively. From 2002 to 2011, he was a Professor and Canada Research Chair at the University of British Columbia (UBC), Vancouver, Canada. Since January 2012 he is an Alexander von Humboldt Professor and the Chair for Digital Communication at FAU. His research interests fall into the broad areas of Communication Theory, Wireless and Molecular Communications, and Statistical Signal Processing.
 
Robert received several awards for his work including the 2002 Heinz Maier Leibnitz Award of the German Science Foundation (DFG), the 2004 Innovations Award of the Vodafone Foundation for Research in Mobile Communications, a 2006 UBC Killam Research Prize, a 2007 Wilhelm Friedrich Bessel Research Award of the Alexander von Humboldt Foundation, the 2008 Charles McDowell Award for Excellence in Research from UBC, a 2011 Alexander von Humboldt Professorship, a 2012 NSERC E.W.R. Stacie Fellowship, a 2017 Wireless Communications Recognition Award by the IEEE Wireless Communications Technical Committee, and the 2022 IEEE Vehicular Technology Society Stuart F. Meyer Memorial Award. Furthermore, he received numerous Best Paper Awards for his work including the 2022 ComSoc Stephen O. Rice Prize and the 2023 ComSoc Leonard G. Abraham Prize. Since 2017, he has been listed as a Highly Cited Researcher by the Web of Science. Robert is a Fellow of the Canadian Academy of Engineering, a Fellow of the Engineering Institute of Canada, and a Member of the German National Academy of Science and Engineering.
 
He served as Editor-in-Chief of the IEEE TRANSACTIONS ON COMMUNICATIONS, VP Publications of the IEEE Communication Society (ComSoc), ComSoc Member at Large, and ComSoc Treasurer. Currently, he serves as Senior Editor of the PROCEEDINGS OF THE IEEE and as ComSoc President.
\end{IEEEbiography}

\end{document}